\documentclass{article}
\usepackage[T2A]{fontenc}
\usepackage[utf8]{inputenc}
\usepackage{hyperref}
\usepackage{amsmath}
\usepackage{amsthm}
\usepackage{amsfonts}
\usepackage{amssymb}
\usepackage[russian,english]{babel}
\usepackage{color}
\usepackage{colortbl}
\usepackage{geometry}
\usepackage[normalem]{ulem}
\usepackage{multicol}
\usepackage[dvipsnames]{xcolor}
\usepackage{tabularx}
\usepackage[table]{xcolor}
\usepackage{amsmath}
\usepackage{graphicx}
\usepackage{tikz}
\usetikzlibrary{arrows.meta, calc}

\theoremstyle{plain}
\newtheorem{theorem}{Theorem}
\newtheorem{proposition}{Proposition}
\newtheorem{lemma}{Lemma}
\newtheorem{claim}{Claim}
\newtheorem{corollary}{Corollary}

\theoremstyle{definition}
\newtheorem{definition}{Definition}

\theoremstyle{remark}

\title{On abelian periodicity of purely morphic words}
\author{Arina Filimonova, Svetlana Puzynina
\\
Saint Petersburg State University, Russia\\
arina4filimonova@gmail.com, s.puzynina@gmail.com}

\begin{document}

\date{}
\maketitle


\begin{abstract}

Deciding periodicity of infinite words generated by morphisms is a classical result in combinatorics on words from 80's by Harju, Linna and Pansiot. In this paper, we are interested in this question in the abelian setting. Two words are called \textit{abelian equivalent} if they contain the same numbers of occurrences of each letter. An infinite word $s$ is called \emph{ultimately abelian periodic} if it can be factorized as $s=uv_1v_2v_3\cdots$, where $v_i$'s are abelian equivalent words.  If $u$ is empty, then $s$ is called \emph{purely abelian periodic}. We provide the following characterization of binary morphisms generating abelian periodic words: A word generated by a binary morphism $f$ is abelian periodic if and only if either it is periodic or there exist an integer $K$ and words $u$, $v$, $u'$, $v'$ such that $f^K(a) = uv$, $f^K(b) = u'v'$, $u\sim_{ab} u'$, and $vu$ and $v'u'$ are abelian periodic with abelian equivalent periods. For the case of the purely abelian periodic words, we also provide an upper bound on $K$ which makes the obtained characterization algorithmic.
\end{abstract}

\section*{Introduction}
\label{sec:intro}
\addcontentsline{toc}{section}{\nameref{sec:intro}}

In this paper, we investigate an abelian analog of periodicity property for the class of words generated by morphisms. Two finite words are  \textit{abelian equivalent} if they are permutations of each other. In other words, in abelian combinatorics on words we consider commutative images of words, so that the order of letters is not taken into account. Various classical results and notions of combinatorics on words admit abelian analog, such as avoidability problems, complexity, powers and others. For a recent survey on abelian properties of words we refer to \cite{AbSurv}. In this paper, we analyze abelian periodicity of infinite words: an infinite word is said to be \textit{abelian periodic} if it can be split into an infinite concatenation of abelian infinite words, where a preperiod is allowed.

A \textit{morphism} is a map defined on the set of words over a given alphabet preserving concatenation. Suppose that a morphim $f$ is non-erasing (i.e., the image of a non-empty word is never an empty word) and is prologable on some letter $a$ (i.e., the image $a$ begins with $a$ and has length at least 2). Then we can generate an infinite word by iterating this morphism as follows: consider the sequence of finite words $a,\;f(a),\;f^2(a),\dots$ and take its limit in the prefix sense. 

Generating infinite words by morphisms is one of the main words constructions in the theory of words. Generating words by morphisms is videly used, for example, in the theory of avoidance, both classical pattern and repetition avoidance and its abelian counterpart (see, e.g., \cite{AVOI, AVO,RR16}). Words generated by morphisms possess rich combinatorial structure and various algebraic and dynamical properties inherited from the morphism. Various properties of words have been studied for the family of words generated by morphisms, and often the situation is more involved in the abelian setting. For example, the complexity of words generated by morphisms has been completely characterized (see \cite{Pansiot84} and references therein), while only partial results are known on their abelian complexity \cite{BFR14},\cite{Wh19}.

Deciding periodicity of words generated by morphisms is a classic result in combinatorics on words (\cite{PER} and independently \cite{PERI}). In this paper, we study the abelian counterpart of this question. Namely, the main question we are interested in is the following: characterize morphisms generating abelian periodic infinite words.

Abelian periodicity property is closely related with several other properties of words and morphisms, which can be expressed in terms of the matrix of a morphism. Given a morphism $f$ on words on an alphabet $\Sigma = \{a_1, a_2,\dots a_n\}$, we associate with it a matrix $M$ of size $n\times n$ with coefficients $m_{ij} = |f(a_j)|_{a_i}$, where $|f(a_j)|_{a_i}$ denotes the number of occurrences of the letter $a_i$ in the word $f(a_j)$. Frequencies of letters in a word generated by a morphism, if they exist, are expressed in terms of eigenvectors of this matrix \cite{PL}. It is straightforward that an abelian periodic word has rational frequencies of letters. So, if a word generated by a morphism has irrational frequences of some letters, then it cannot be abelian periodic; however, rationality of frequences is not sufficient for abelian periodicity.

Another property that could be partially characterized via eigenvalues of the matrix of a morphism is $c$-balance of a purely morphic word \cite{BAL}. An infinite word is called $c$-balanced if for each letter $a$ the numbers of occurrences of $a$ in any two factors of the same length differ by at most $c$. 
The $c$-balance property if also a necessary but not a sufficient condition for abelian periodicity.

To formulate the main results of this paper, we will need the eigenvalue of the matrix of a morphism with the second largest absolute value; we denote it by $\theta_2$.
The following theorem gives a characterization of the abelian periodicity of words generated by binary morphisms:

\begin{theorem}
\label{th:main}
    Let $f$ be a nonerasing binary morphism prolongable on $a$ with matrix $M_f$.
    
    \begin{itemize}
    
    \item If $f$ is primitive with $\theta_2 \neq 0$, then the word $f^{\omega}(a)$ is abelian periodic if and only if $f$ is of the form $f(a) = a(ba)^k$, $f(b) = b(ab)^m$ for some integers $k$ and $m$.

    \item If $f$ is primitive with $\theta_2 = 0$, then the word  $f^{\omega}(a)$ is abelian periodic if and only if there exist an integer $K$ and  words $u$, $v$, $u'$, $v'$ such that $f^K(a) = uv$, $f^K(b) = u'v'$, $u\sim_{ab} u'$, and $vu$ and $v'u'$ are abelian periodic with abelian equivalent periods.

    \item If $f$ is non-primitive, then the word $f^{\omega}(a)$ is abelian periodic if and only if it is periodic.

    \end{itemize}
\end{theorem}

We note that the criterion for the case $\theta_2=0$ is not algorithmic, since we have no upper bound on $K$ so far. However, for purely abelian periodic words, we prove an upper bound and hence we have a criterion that could be checked algorithmically. Note that for a binary morphism $f$ with $\theta_2=0$, its matrix $M_f$ can be written in the form 
$\begin{pmatrix}
n A & m A\\
n B & m B
\end{pmatrix}$ with $m,\,n \in \mathbb{N}$ and $(m,\;n) = 1$ and for some integers $A$ and $B$. 


\begin{theorem}[Upper bound on $K$ in case of pure abelian periodicity]
\label{th:bound}
Let $f$ be a morphism such that its matrix $M_f$ has the second eigenvalue $\theta_2 = 0$, i. e., $M_f$ is of the form
$\begin{pmatrix}
n A & m A\\
n B & m B
\end{pmatrix}$ for some integers $m,n, A, B$ such that $(m,\;n) = 1$. If  the word $f^{\omega}(a)$ is purely abelian periodic, then the following upper bound holds for $K$ from Theorem \ref{th:main}:
$$K \leq (|f(a)| + |f(b)|)^{m+n-2}.$$
\end{theorem}

In Section \ref{section:secArbitrary1}, we consider primitive morphisms with $\theta_2 \neq 0$ and provide a proof of the first part of Theorem \ref{th:main}. In Section \ref{section:secArbitrary2}, we study primitive morphisms with $\theta_2 = 0$, complete the proof of Theorem \ref{th:main} and provide a proof of Theorem \ref{th:bound}. In particular, in Subsection \ref{subsec:ab_bbaa}, we consider an example of a morphism $f:a\mapsto ab, b \mapsto bbaa$, and prove that the word $f^{\omega}(a)$ is not abelian periodic. In Subsection \ref{subsec:parikh_coll}, we study the general case $\theta_2 = 0$ and complete the proof of Theorem \ref{th:main}. In Subsection \ref{subsection:upper_bound}, we prove Theorem \ref{th:bound} and give a corollary on the decidability of pure abelian periodicity of words generated  by morphisms.

\section{Preliminaries and Notation}

In this section, we recall several relevant results from the literature and introduce the notation used in the paper.

An \textit{alphabet} $\Sigma$ is a finite set; it's elements are called \textit{letters}. A \textit{word} is a finite or infinite sequence of letters. We denote by $\Sigma^*$ the set of all finite words over the alphabet $\Sigma$. Furthermore, we define $\Sigma^+ = \Sigma^*\setminus \varepsilon$, where $\varepsilon$ is an empty word. The \emph{length} $|u|$ of a finite word $u$ is the number of letters in it.

A \emph{concatenation} is a binary operation on words 
defined as follows: for any two words $u = u_0u_1\cdots u_n$ and $v = v_0v_1\cdots v_m$, their concatenation $uv$ is the word $u_0u_1\cdots u_nv_0v_1\cdots v_m$. 
A word $u$ is called a \emph{factor} of a word $s$ if there exist (possibly empty) words $x$ and $y$ such that $s = xuy$. If $x = \varepsilon$, the factor $u$ is called a \emph{prefix}; if $y = \varepsilon$, it is called a \emph{suffix}.
We let pref$_k(s)$ denote the prefix of length $k$ of a word $s$.

A mapping $f: \Sigma^* \rightarrow \Sigma^*$ is called a \emph{morphism} if it preserves concatenation; that is, for each pair of words $u$, $v$ we have $f(uv) = f(u)f(v)$. It is easy to see that a morphism $f: \Sigma^* \rightarrow \Sigma^*$ is completely defined by its images on letters of $\Sigma$. Thus, any morphism can be naturally extended to the set of all infinite words over $\Sigma$ (we denote this set by $\Sigma^\mathbb{N}$). The images of letters under a morphism $f$ are called \emph{blocks}.
We only consider \emph{nonerasing} morphisms, i.e., morphisms for which $f(a) \neq \varepsilon$ for each letter $a \in \Sigma$.
The \emph{length} $|f|$ of the morphism $f: \Sigma^* \rightarrow \Sigma^*$ is the sum of the lengths of its blocks, i.e., $|f| = \sum\limits_{a\in \Sigma} |f(a)|.$ A nonerasing morphism $f$ is called \emph{prolongable on} $a\in \Sigma$ if $f(a) = au$ for some $u \in \Sigma^+$. Note that in this case, for every $k \in \mathbb{N}$, the word $f^k(a)$ is a proper prefix of $f^{k+1}(a)$. Thus, there exists a unique infinite word $f^{\omega}(a)$ such that for every $k \in \mathbb{N}$ the word $f^k(a)$ is its prefix; such words $f^{\omega}(a)$ are called \textit{purely morphic words}, or \textit{words generated by morphisms}. Note that $f(f^{\omega}(a)) = f^{\omega}(a)$, i.e., the word $f^{\omega}(a)$ is a fixed point of the morphism $f$. 

Consider $a\in \Sigma$ and $u \in \Sigma^*$. Denote by $|u|_a$ the number of occurrences of the letter $a$ in the word $u$.
Two finite words $u$ and $v$ over $\Sigma$ are called \emph{abelian equivalent} if for each letter $a\in \Sigma$ we have $|u|_a = |v|_a$. An infinite word $s = s_0s_1s_2\cdots,$ with $s_i \in \Sigma$, is called \emph{eventually abelian periodic} with period $p$ and preperiod $r$ if finite words $s_rs_{r+1}\cdots s_{r+p-1}$ and $s_{r + kp}s_{r + kp + 1}\cdots s_{r+(k+1)p-1}$ are abelian equivalent for each $k \in \mathbb{N}$ (if $r=0$, $s$ is called \emph{purely abelian periodic}). We will shorten ''eventually abelian periodic'' to just ''abelian periodic''. Whenever we refer to pure abelian periodicity, we specify it each time.

An \textit{abelian complexity} of an infinite word is a function counting, for each integer $n$, the number of abelian classes of factors of length $n$ of the word \cite{COM}. It is easy to see that if a word is abelian periodic, then its abelian complexity is bounded (the converse is not true in general).

A morphism $f: \Sigma^* \rightarrow \Sigma^*$ is called \emph{uniform} if for each $a$ and $b$ in $\Sigma$ we have $|f(a)| = |f(b)|$. A \textit{coding} is a uniform morphism with block length equal to $1$. A uniform  morphism $f$ with $|f(a)| = d$ is called \emph{bijective} if for every $j \in \{0,\,1,\dots,\,d-1\}$ the mapping $a \mapsto f(a)_j$ (the letter $a \in \Sigma$ maps to the $j$-th letter of the word $f(a)$) is a bijection on $\Sigma$. For example, the following morphism $h$ over the alphabet $\{1,\;2,\;3,\;4,\;5,\;6\}$ is bijective: $$\begin{cases}
            h(1) = 123,\\
            h(2) =  456,\\
            h(3) =  345,\\
            h(4) =  634,\\
            h(5) =  561,\\
            h(6) =  212.
        \end{cases}$$
        
Let $k\in \mathbb{N}$. An infinite word $s$ on an alphabet $\Sigma$ is called \emph{$k$-automatic} if there exists a deterministic finite automaton $\mathcal{A} = (Q,\, \{0,\;1,\dots, k-1\},\, \delta,\, q_0 )$ and a mapping $\tau: Q\rightarrow \Sigma$, such that $$s_n = \tau(\delta(q_0, [n]_k)),$$ where $[n]_k$ is the base-$k$ representation of $n$ \cite{AS}. 

Let $f$ be a morphism on an alphabet $\Sigma = \{a_1, a_2,\dots a_n\}$. We associate with $f$ an $n \times n$ matrix $M_f = (m_{ij})$, where the element $m_{ij} = |f(a_j)|_{a_i}$ is the number of occurrences of the letter $a_i$ in the block $f(a_j)$. For example, let $\Sigma = \{a, b\}$. Then the morphism $f(a) = aab,\, f(b) = bbaab$ has the associated matrix 
$\begin{pmatrix}
2 & 2\\
1 & 3
\end{pmatrix}$. We denote the eigenvalues of the matrix $M_f$ by $\theta_1,\,\theta_2,\,\dots,\,\theta_n$ in non-increasing order of their absolute values, i.e., $|\theta_1| \geq |\theta_2| \geq \dots \geq |\theta_n|$. A morphism $f$ is called \emph{primitive} if there exists an integer $k \in \mathbb{N}$ such all entried of the matrix $(M_f)^k$ are positive. 

Let $a\in \Sigma$. The \emph{frequency of the letter $a$} in an infinite word $s \in \Sigma^{\mathbb{N}}$ is defined as the limit (if it exists) $$\rho_s(a) = \lim\limits_{n \rightarrow \infty} \frac{|s_0s_1\dots s_{n-1}|_a}{n}.$$ If the limit does not exist, the frequency is not defined.
The letter frequencies in a purely morphic word are given by the following well-known result:

\begin{proposition}[Prop. 5.8 в \cite{PL}]
   \label{utv:PL}
   Let $f$ be a primitive morphism prolongable on $a$, and let $\theta_1$ be the eigenvalue of $M_f$ with the largest absolute value.
   Then the vector of letter frequencies in the word $f^{\omega}(a)$ (i.e., the vector with entries $v_i = \rho(a_i)$, with $a_i \in \Sigma$) is given by the eigenvector $v$ corresponding to the eigenvalue $\theta_1$, normalized so that the sum of its entries equals 1.
\end{proposition} 
Such a vector $v$ is unique, as the eigenspace corresponding to the eigenvalue $\theta_1$ of $M_f$ is one-dimensional (see, for example, Theorem 5.4 in \cite{PL}).

An infinite word $s \in \Sigma^{\mathbb{N}}$ is called \emph{$c$-balanced} (for a $c\in \mathbb{N}$) if for any two of its finite subwords $u$ and $v$ of equal length and for every letter $a \in \Sigma$ the inequality $||u|_a - |v|_a| \leq c$ holds. It is easy to see that each abelian periodic infinite word $s$ is $c$-balanced for $c = r + 2p$, where $p$ is the abelian period of $s$ and $r$ is its preperiod. The following criterion can be used to check the $c$-balance of a purely morphic word:

\begin{theorem}[Theorem 13 in \cite{BAL}]
    \label{theorem:BAL}
    Let $v$ be a word generated by a primitive morphism $f$, and let $\theta_2$ be the eigenvalue of the matrix $M_f$ with the second largest absolute value. Then the word $v$ is $c$-balanced for some $c \in \mathbb{N}$ if and only if one of the following holds:
    
    \begin{itemize}
        \item $|\theta_2| < 1$;
        \item $|\theta_2| = 1$, and $\theta_2$ is a root of unity, has algebraic multiplicity $1$ as an eigenvalue of $M_f$, and we have $A_{f, v} = 0$;
    \end{itemize}
    where the complex number $A_{f, v}$ depends on the morphism $f$ and the fixed point $v$ and can be computed explicitly (see Appendix B in \cite{BAL}).
\end{theorem}

\section{Binary Morphisms: The Case $\theta_2 \neq 0$}
\label{section:secArbitrary1}

In this section, we consider binary morphisms $f$ for which their matrix $M_f$ has the second eigenvalue $\theta_2 \neq 0$. This corresponds to the first part of Theorem \ref{th:main}. First note the following:

\begin{lemma}
\label{lemma:lemReal0}
   For a binary morphism $f$, the eigenvalues of $M_f$ are real numbers. 
\end{lemma}

\begin{proof}
    The matrix $M_f$ corresponding to the morphism $f$ is a $2\times2$ matrix with nonnegative entries. It is straightforward to see that the eigenvalues of such a matrix are real.
\end{proof}

\subsection{The Case $0 < |\theta_2| < 1$}

We prove that if $0 < |\theta_2| < 1$, the letter frequencies in the word $f^{\omega}(a)$ are irrational, and thus the word is not abelian periodic.

\begin{lemma}
\label{lemma:lemIrrational}
   For a binary morhism $f$, if $0 < |\theta_2| < 1$, then $\theta_2$ is irrational.
\end{lemma}

\begin{proof} By Lemma \ref{lemma:lemReal0}, the eigenvalues of $M_f$ are real. By rational root theorem, all rational eigenvalues must be integers. Since for $\theta_2$ we have $0 < |\theta_2| < 1$, it must be irrational.
\end{proof}

\begin{corollary}
\label{cor:cor_Irrational}
  For a binary morhism $f$, if $0 < |\theta_2| < 1$, then the frequencies of letters in the word $f^{\omega}(a)$ are irrational.
\end{corollary}

\begin{proof}
    The matrix of the morhpism $f$ is of the form $\begin{pmatrix}
A + \theta_2 & \alpha A\\
B & \alpha B + \theta_2
\end{pmatrix}$ for some $A,\; B,\; \alpha$, such that $|A + \alpha B + \theta_2| \geq |\theta_2|$ (since $\theta_2$ is the second eigenvalue). By Lemma \ref{lemma:lemIrrational}, we have that $\theta_2$ is irrational. So, $A$ is irrational (since $A + \theta_2$ is an integer). Since $B$ is integer, we have that  $\frac{B}{A+B}$ is irrational, and the latter value is the frequency of the letter $b$ in the word $f^{\omega}(a)$ by Proposition \ref{utv:PL}.

\end{proof}

\subsection{The Case $|\theta_2| = 1$}


In this subsection, we treat separately the case $|\theta_2| = 1$. First note that for a binary morphism we have the following:

\begin{lemma}
\label{lemma:lemReal}
   If $|\theta_2| = 1$, then $\theta_2 = \pm 1$. 
\end{lemma}

\begin{proof}
    Follows from Lemma \ref{lemma:lemReal0}.
\end{proof}




So, $\theta_2 = \pm 1$. Note that if $\theta_2=-1$ for $M_f$, then for the matrix $M_{f}^2$ we have $\theta_2=1$. Since $(f^2)^{\omega}(a) = f^{\omega}(a)$, we may assume without loss of generality that for $M_f$ we have $\theta_2=1$. In this case, the morphism $f$ has the matrix$\begin{pmatrix}
A+1 & \alpha A\\
B & \alpha B+1
\end{pmatrix}$ for some integers $A$ and $B$ and for some rational $\alpha$, since $M_f - \theta_2 I$ is a singular matrix
, so its columns are collinear. 

Note that in these cases, the letter frequencies are equal to ($\frac{A}{A+B},\; \frac{B}{A+B}$).

\begin{lemma}
\label{lemma:lem_Arb_1}
    Let $f$ be a binary morphism with matrix $M_f = \begin{pmatrix}
A+1 & \alpha A\\
B & \alpha B+1
\end{pmatrix}$. Then for every word $u \in \{a,\;b\}^*$ the following holds: $B|f(u)|_a - A|f(u)|_b = B|u|_a - A|u|_b$.
\end{lemma}

\begin{proof}
Straightforward.
\end{proof}

\begin{lemma}
\label{lemma:lem2.2}
Suppose that the frequencies of  both letters in a word $s\in \{a,b\}^{\mathbb{N}}$ exist and are equal to $(\frac{A}{A+B}, \frac{B}{A+B})$. If for each $M\in\mathbb{N}$ there exists a factor $u_M$ of the word $s$ such that $$B|u_M|_a - A|u_M|_b \geq M,$$ then the word $s$ is not $c$-balanced for any $c \in \mathbb{N}$.
\end{lemma}

\begin{proof}

First, we note that since the letter frequencies in the word $s$ are $(\frac{A}{A+B}, \frac{B}{A+B})$, it follows that  
\begin{equation}\label{eq:lim}\lim\limits_{k\rightarrow \infty}\frac{B|\text{pref}_k s|_a}{A|\text{pref}_k s|_b} = 1.\end{equation}

Assume that the word $s$ is $c$-balanced for some $c\in \mathbb{N}$. Consider $M \geq (A+B) c+1$. Let $u_M$ be a factor of the word $s$ such that $B|u_M|_a - A|u_M|_b \geq M$. Then, by the $c$-balance of the word $s$, for each factor $v$ of $s$ with $|v| = |u_M|$ the following holds: 
\[B|v|_a - A|v|_b \geq B(|u_M|_a - c) - A(|u_M|_b + c) = (B|u_M|_a - A|u_M|_b) - (A+B)c \geq M - (A+B)c \geq 1.\]

We now factorize the word $s$ into words of length $|u_M|$ and observe that for each integer $r$ and for a prefix $t_r$ of length $r|u_M|$ of $s$ the following holds: $\frac{B|t_r|_a}{A|t_r|_b} \geq 1 + \frac{r}{A|t_r|_b} \geq 1 + \frac{r}{Ar(|u_M|_b+c)} = 1 + \frac{1}{A(|u_M|_b+c)}$. This contradicts \eqref{eq:lim}.
\end{proof}

\begin{lemma}
\label{lemma:lem_Arb_aa}
Let $f$ be a morphism with matrix $\begin{pmatrix}
A+1 & \alpha A\\
B & \alpha B+1
\end{pmatrix}$. If $f(a)$ contains $aa$ as a factor, then for each 
integer $c$ the word $f^{\omega}(a)$ is not $c$-balanced. In particular, the word $f^{\omega}(a)$ is not abelian periodic.
\end{lemma}

\begin{proof}
        Suppose that the image of the letter $a$ contains the word $aa$, i.e, $f(a) = uaav$ for some $u,v \in \Sigma^*$. Note that, by assumption, $B|f(a)|_a - A|f(a)|_b = B \geq 1$. Thus, we have either $B|ua|_a - A|ua|_b \geq 1$, or $B|av|_a - A|av|_b \geq 1$. Without loss of generality assume that $B|ua|_a - A|ua|_b \geq 1$. We will prove that for each integer $M$ there exists a factor $t_M$ of $f^{\omega}(a)$ such that $$B|t_M|_a - A|t_M|_b \geq M.$$

    We construct the sequence $ (t_M)_{M\geq 1}$ inductively. First, consider the word $t_1 = ua$: by assumption, $B|ua|_a - A|ua|_b \geq 1$. Note that $f^{\omega}(a)$ contains an occurrence of $t_1$ followed by the letter $a$ (that is, an occurrence of $uaa$).

    Now, consider the word $t_2 = f(ua)ua$. Note that it occurs in the word $f^{\omega}(a)$, for instance, as a prefix of the word $f(uaa)$; moreover, this occurrence of $t_2$ is followed by the letter $a$.

    Next, consider the occurrence of the word $t_{k+1} = f(t_k)ua$ as a prefix of an occurrence of $f(t_ka)$. This occurrence of $t_{k+1}$ is followed by the letter $a$, thus, there is an occurrence of $t_{k+1}a$ in the word $f^{\omega}(a)$.

    It is easy to show by induction that $B|t_k|_a - A|t_k|_b \geq k$. Indeed, for $k=1$ the equality holds, and for $k > 1$ we have $t_k = f(t_{k-1})ua$, where by Lemma \ref{lemma:lem_Arb_1} we have $B|f(t_{k-1})|_a - A|f(t_{k-1})|_b  = B|t_{k-1}|_a - A|t_{k-1}|_b \geq k - 1$, and $B|ua|_a - A|ua|_b \geq 1$. Applying Lemma \ref{lemma:lem2.2}, we get that $f^{\omega}(a)$ is not $c$-balanced for any $c\in\mathbb{N}$, and hence  is not abelian periodic.
\end{proof}

\begin{lemma}
\label{lemma:lem_Arb_bb}
    Let $f$ be a morphism with matrix $M_f = \begin{pmatrix}
A+1 & \alpha A\\
B & \alpha B+1
\end{pmatrix}$. If $f(b)$ contains $bb$ as a factor, then for each $c\in\mathbb{N}$ the word $f^{\omega}(a)$ is not $c$-balanced. In particular, the word $f^{\omega}(a)$ is not abelian periodic.
\end{lemma}

\begin{proof}
    Similarly to the proof of Lemma \ref{lemma:lem_Arb_aa}.
\end{proof}

\begin{corollary}
\label{cor:cor_Arb_1}
Let $f$ be a morphism with matrix $M_f = \begin{pmatrix}
A+1 & \alpha A\\
B & \alpha B+1
\end{pmatrix}$. If either $f(a)$ or $f(b)$ contains an occurrence of $aa$ or an occurrence of $bb$, then the word $f^{\omega}(a)$ is not $c$-balanced for any $c\in\mathbb{N}$ and, therefore, is not abelian periodic.
\end{corollary}

\begin{proof}
     If either $f(a)$ or $f(b)$ contains an occurrence of $aa$ or an occurrence of $bb$, then for the morphism $f^2$ we have an occurrence of $aa$ or an occurrence of $bb$ in both $f^2(a)$ and $f^2(b)$. It then remains to apply either Lemma \ref{lemma:lem_Arb_aa} or Lemma \ref{lemma:lem_Arb_bb}.
\end{proof}

\begin{corollary}
\label{cor:cor_Arb_2}
Let $f$ be a morphism with matrix $M_f = \begin{pmatrix}
A+1 & \alpha A\\
B & \alpha B+1
\end{pmatrix}$. Then the word $f^{\omega}(a)$ is abelian periodic if and only if $A = B$ and $f(a) = a(ba)^A$, $f(b) = b(ab)^{\alpha A}$.
\end{corollary}

\begin{proof}
    If the difference between the numbers of occurrences of letters $a$ and $b$ in the image of some letter is at least $2$, then the image of this letter contains two consecutive occurrences of the same letter, which contradicts Corollary \ref{cor:cor_Arb_1}. Thus, $A = B$.
\end{proof}

\begin{corollary}
\label{cor:cor_Arb_3}
If $f$ is a morphism with matrix $M_f = \begin{pmatrix}
A-1 & \alpha A\\
B & \alpha B-1
\end{pmatrix}$, then the word $f^{\omega}(a)$ is not abelian periodic.
\end{corollary}

\begin{proof}
For the morpfism $f^2$, the eigenvalue $\theta_2$ of the matrix $M_{f^2}$ is equal to 1.
By Corollary \ref{cor:cor_Arb_1}, we obtain the following necessary condition for abelian periodicity of the word $f^{\omega}(a) = (f^2)^{\omega}(a)$: there must be no occurrences of $aa$ and $bb$ in both $f^2(a)$ and $f^2(b)$.
So, the difference between the numbers of occurrences of letters $a$ and $b$ must be at most $1$ in both $f(a)$ and $f(b)$, otherwise either $f(a)$ or $f(b)$ contain an occurrence of $aa$ or $bb$.  
If $B>A$, then $|f(a)|_b\geq |f(a)|_a+2$. If $B<A$, then $\alpha B < \alpha A$ and $|f(b)|_a\geq |f(a)|_b+2$. In the case $A=B$, i. e., $M_f = \begin{pmatrix}
A-1 & \alpha A\\
A & \alpha A-1
\end{pmatrix}$, 
the block $f(a)$ contains an occurrence of $bb$, since $f(a)$ begins with the letter $a$.
\end{proof}

Thus, we have established the following result constituting the first part of Theorem \ref{th:main}:

\begin{proposition}
\label{utv:utv_Arb_1}
    Let $f$ be a primitive binary morphism such that for its matrix $M_f$ we have $\theta_2 \neq 0$. Then the word $f^{\omega}(a)$ is abelian periodic if and only if the morphism $f$ is of the form $f(a) = a(ba)^k$, $f(b) = b(ab)^m$ for some $k, \, m\in \mathbb{N}$.
\end{proposition}

\begin{proof}
    If the morphism $f$ is of the form $f(a) = a(ba)^k$, $f(b) = b(ab)^m$, the abelian periodicity of the word $f^{\omega}(a)$ is obvious: the word is periodic in the classic sence, so it is also periodic in the abelian sence.

    It remains to prove that is not of this form, the abelian periodicity does not hold. Consider the eigenvalue $\theta_2$ of the matrix $M_f$.
    
    If $|\theta_2| > 1$, then the word $f^{\omega}(a)$ is not $c$-balanced for any $c$ (see Theorem \ref{theorem:BAL}), and hence it is not abelian periodic.
    
    If $0 < |\theta_2| < 1$, then, by Corollary \ref{cor:cor_Irrational}, the letter frequencies in the word $f^{\omega}(a)$ are irrational, and hence $f^{\omega}(a)$ is not abelian periodic.
    
    If $|\theta_2| = 1$, then, by Lemma \ref{lemma:lemReal}, we have $\theta_2 = \pm 1$. Applying Corollaries \ref{cor:cor_Arb_2} and \ref{cor:cor_Arb_3}, we obtain the desired result.
\end{proof}

    \subsection{Corollary: a criterion of the abelian periodicity for uniform binary morphisms}

The following theorem gives a characterization of the abelian periodicity of words generated by uniform binary morphisms. 

 \begin{theorem}
\label{theorem:thm1}
    Let $f$ be a uniform morphism prolongable on $a$ over the binary alphabet $\Sigma = \{a,\,b\}$. Then the word $f^{\omega}(a)$ is abelian periodic if and only if one of the following conditions holds:
    \begin{enumerate}
    
        \item The eigenvalue $\theta_2 = 0$ for the matrix $M_f$, i.e., $M_f = \begin{pmatrix}
A & A\\
B & B
\end{pmatrix};$

\item The morphism $f$ is of the form $f(a) = a(ba)^k$ and $f(b) = b(ab)^k$ for some $k\in \mathbb{N}$.

    \end{enumerate}
\end{theorem}

\begin{proof}
    The abelian periodicity of $f^{\omega}(a)$ is obvious in the listed cases. Indeed, in Case 1 the blocks $f(a)$ and $f(b)$ are abelian equivalent, so $|f(a)| = |f(b)|$ is an abelian period of the word $f^{\omega}(a)$. In case 2 the word $f^{\omega}(a)$ is equal to $ (ab)^{\omega}$, so it is (abelian) periodic with period $2$.

    It remains to prove that in all other cases there is no abelian periodicity. Consider the second eigenvalue $\theta_2$ of the morphism $f$. Note that for binary uniform morphisms we have $\theta_2 \in \mathbb{Z}.$
    
    If $|\theta_2| > 1$, then the word $f^{\omega}(a)$ is not $c$-balanced for any $c$ (see Theorem \ref{theorem:BAL}), and so it is not abelian periodic.
    
    If $|\theta_2| < 1$, then in the binary case we have $\theta_2 = 0$. 
    
    If $\theta_2 = 1$, then the matrix is of the form          $\begin{pmatrix}
A+1 & A\\
B & B+1
\end{pmatrix}$. By Corollary \ref{cor:cor_Arb_2}, the word $f^{\omega}(a)$ is abelian periodic if and only if $A = B$ and $f(a) = a(ba)^A$, $f(b) = b(ab)^A$.

Finally, suppose that $\theta_2 = -1$, i.e., we have  $M_f=\begin{pmatrix}
A-1 & A\\
B & B-1
\end{pmatrix}$. By Corollary \ref{cor:cor_Arb_3}, in this case the word $f^{\omega}(a)$ is not abelian periodic.
\end{proof}

\section{Binary morphisms: the case $\theta_2 = 0$}
\label{section:secArbitrary2}

For non-uniform morphisms the case $\theta_2 = 0$ is non-trivial: among binary morhisms with $\theta_2 = 0$ there exist morhisms generating abelian periodic words, as well as morphisms generating non-abelian aperiodic words. We start with an example of a morphismgenerating a word which is not eventually abelian periodic.

\subsection{Example: morphism $a \mapsto ab$, $b \mapsto bbaa$}
\label{subsec:ab_bbaa}

In this subsection, we consider the following morphism
 $f: \{a,\,b\}^* \rightarrow \{a,\,b\}^*$:
$$\begin{cases}
    f(a) = ab, \\
    f(b) = bbaa.
\end{cases}$$

We will show that the word generated by this morphism is not abelian periodic, and the proof is surprisingly involved for a word having fairly clear structure.

\begin{proposition}
\label{pr:abel1}
    A word $f^{\omega}(a)$ generated by the following morphism $$\begin{cases}
        f(a) = ab\\
        f(b) = bbaa
    \end{cases}$$
is not (eventually) abelian periodic.
\end{proposition}

For the proof of this proposition we will need a series of auxiliary statements. First we show that this word is so-called hidden automatic \cite{HidAut}, i.e., although it is generated by a non-uniform morphism, it also admits a definition as an automatic word. In fact, this is implicitly proved in \cite{DEK1}; however, we provide a proof here, since we will make use of the construction.

Consider the following morphism
 $h: \{1, 2, \dots, 6\} \rightarrow  \{1, 2, \dots, 6\}$:         
 
$$\begin{cases}
            h(1) = 123\\
            h(2) =  456\\
            h(3) =  345\\
            h(4) =  634\\
            h(5) =  561\\
            h(6) =  212.
        \end{cases}$$

    \begin{lemma}
    \label{lemma:l1}
        We have $f^{\omega}(a)$ = $\tau(h^{\omega}(1))$, where  $\tau: \{1, 2, \dots, 6\} \rightarrow \{a, b\}$ is a coding defined as follows:     
        $$\begin{cases}
            \tau(1) = \tau(5) = \tau(6) = a\\
            \tau(2) = \tau(3) = \tau(4) = b.
        \end{cases}$$
        
    \end{lemma}

    \begin{proof}
        The word $f^{\omega}(a)$ is a fixed point of the morphism $f$, so we can consider a natural factorization of $f^{\omega}(a) = f(f^{\omega}(a))$ into blocks $f(a)$ and $f(b)$. We now associate the block $f(a) = ab$ with the word $12$, and the block $f(b) = bbaa$ with the word $3456$; we denote this correspondence by $\sim$, i.e., $f(a) \sim 12$, $f(b) \sim 3456$. Thus, any word that is a concatenation of blocks, i.e. a word of the form $f(w)$, is associated with some word on the alphabet $\{1, 2, \dots, 6\}$ (note that the word of the form $f(w)$ is uniquely divided into blocks). Note also that if $f(w) \sim s$, $s\in \{1, 2, \dots, 6\}^*$, then $f(w) = \tau(s)$. In the same way, the word $f^{\omega}(a)$ is associated with an infinite word on the alphabet $\{1, 2, \dots, 6\}$. It remains to prove that the latter word equals $h^{\omega}(1)$.

        We now prove that for each $m\in \mathbb{N}$ we have $f^{m}(ab) \sim h^{m}(12)$ and $f^{m}(bbaa) \sim h^{m}(3456)$.
        Then $f^{\omega}(a) = f^{\omega}(ab) = \tau(h^{\omega}(12)) = \tau(h^{\omega}(1))$. Here, $f^{\omega}(ab)$ is the word generated by $f$ starting with the word $ab$ in the same sense as $f^{\omega}(a)$ is generated by $f$ starting with the letter $a$, since $f^m(ab)$ is a prefix of $f^{m+1}(ab)$ for each integer $m$. Note that $f(a) = ab$, so, clearly, $f^{\omega}(a) = f^{\omega}(ab)$); the same applies to $h^{\omega}(1)$. 
        The proof is by induction on $m$.
        
        For $m = 1$ the statement is verified directly:\begin{align*}  f(ab) = abbbaa,\; & h(12) = 123456,\\  f(bbaa) = bbaabbaaabab,\; & h(3456) = 345634561212.        \end{align*} Now suppose that $m>1$; then $$h^m(12) = h^{m-1}(h(12)) = h^{m-1}(123456) = h^{m-1}(12)h^{m-1}(3456) \sim $$  $$\sim f^{m-1}(ab)f^{m-1}(bbaa) = f^m(ab).$$ Similarly, we have $$h^m(3456) = h^{m-1}(3456)h^{m-1}(3456)h^{m-1}(12)h^{m-1}(12) \sim $$ $$\sim f^{m-1}(bbaa)f^{m-1}(bbaa)f^{m-1}(ab)f^{m-1}(ab) = f^m(bbaa).$$
        \end{proof}


        The following classical result, known as Conham's theorem, provides an equivalent definition of automatic words via uniform morphisms:

        \begin{theorem}[\cite{COBH}, see also Theorem 6.3.2 in \cite{AS}]
            \label{theorem:Cobh}
            An infinite word $v$ on an alphabet  $\Gamma$ is $k$-automatic if and only if $v = \tau(g^{\omega}(a))$, where $g: \Sigma^* \rightarrow \Sigma^*$ is a uniform morphism with block length equal to $k$, and $\tau: \Sigma \rightarrow \Gamma$ is a coding.
        \end{theorem}

         Note that, given a uniform morphism $g$, the corresponding automaton is built algorithmically, and vise versa.

        Lemma \ref{lemma:l1} and Theorem $\ref{theorem:Cobh}$ imply the following:

    \begin{lemma}
        The word $f^{\omega}(a)$ is 3-automatic.
    \end{lemma}
    
   The corresponding automation is illustrated in Fig.  \ref{fig:mpr}.

   \begin{figure}[h]
    \centering

\resizebox{\textwidth}{!}{
\begin{tikzpicture}
\begin{scope}[every node/.style={circle,thick,draw,minimum size=1cm,font=\LARGE,inner sep=2pt}]
    \node (A) at (-5,6) {1};
    \node (F) at (-5,0) {6};
    \node (B) at (0,8) {2};
    \node (D) at (5,0) {4};
    \node (E) at (0,-2) {5};
    \node (C) at (5,6) {3};
\end{scope}

\begin{scope}[>={Stealth[black]},
              every node/.style={circle, font=\Large, inner sep=2.5pt},
              every edge/.style={draw=black,very thick}]
    
    \path [->] (A) edge node[above, pos=0.4] {$1$} (B);
    \path [->] (A) edge node[above, pos=0.8] {$2$} (C);
    \path [->] (B) edge node[below left, pos=0.4, xshift=2pt] {$0$} (D);
    \path [->] (B) edge node[right, pos=0.4, xshift=-4pt] {$1$} (E);
    \path [->] (B) edge node[below right, pos=0.4, xshift=-3pt] {$2$} (F);
    
    \path [->] (C) edge node[right, pos=0.5, xshift=-4pt] {$1$} (D);
    \path [->] (C) edge node[below right, pos=0.7] {$2$} (E);

    \path [->] (D) edge[bend left = 90, looseness = 1.5] node[below, pos=0.5] {$0$} (F);
    \path [->] (D) edge[bend right = 45] node[right, pos=0.5, xshift=-4pt] {$1$} (C);

    \path [->] (E) edge node[below, pos=0.5] {$1$} (F);
    \path [->] (E) edge node[below left, pos=0.4] {$2$} (A);

    \path [->] (F) edge node[left, pos=0.5, xshift=2pt] {$1$} (A);

    \path [->] (F) edge[bend left = 90, looseness = 1.9] node[left, pos=0.4] {$0, 2$} (B);
    
    \path [->] (A) edge[out=180, in=90, looseness=8] node[above, pos=0.5] {$0$} (A);
    \path [->] (C) edge[out=90, in=0, looseness=8] node[above, pos=0.5] {$0$} (C);
    \path [->] (D) edge[out=15, in=-75, looseness=8] node[below, pos=0.5, xshift=5pt, yshift=5pt] {$2$} (D);
    \path [->] (E) edge[out=-45, in=-135, looseness=8] node[below, pos=0.5] {$0$} (E);
\end{scope}
\end{tikzpicture}
}

   \caption{Automaton corresponding to the morphism $h$.}
    \label{fig:mpr}
    \end{figure}
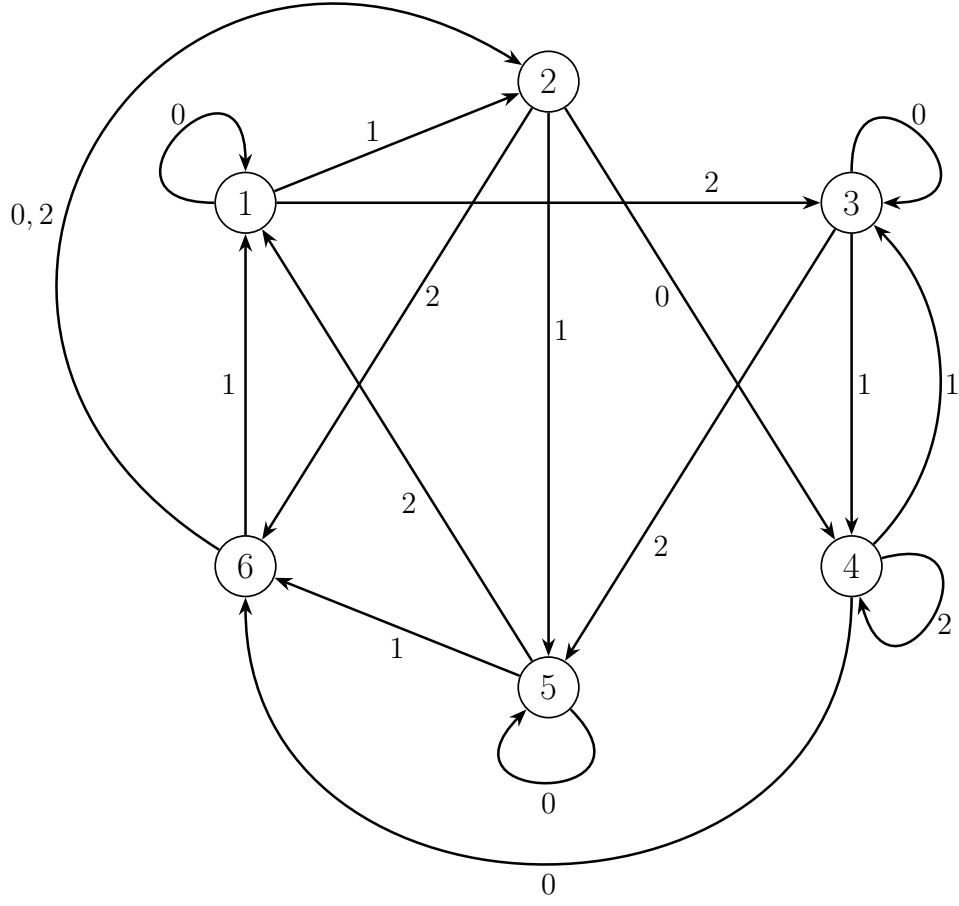

    \begin{figure}[h]
    \centering
    \includegraphics[width=1\linewidth]{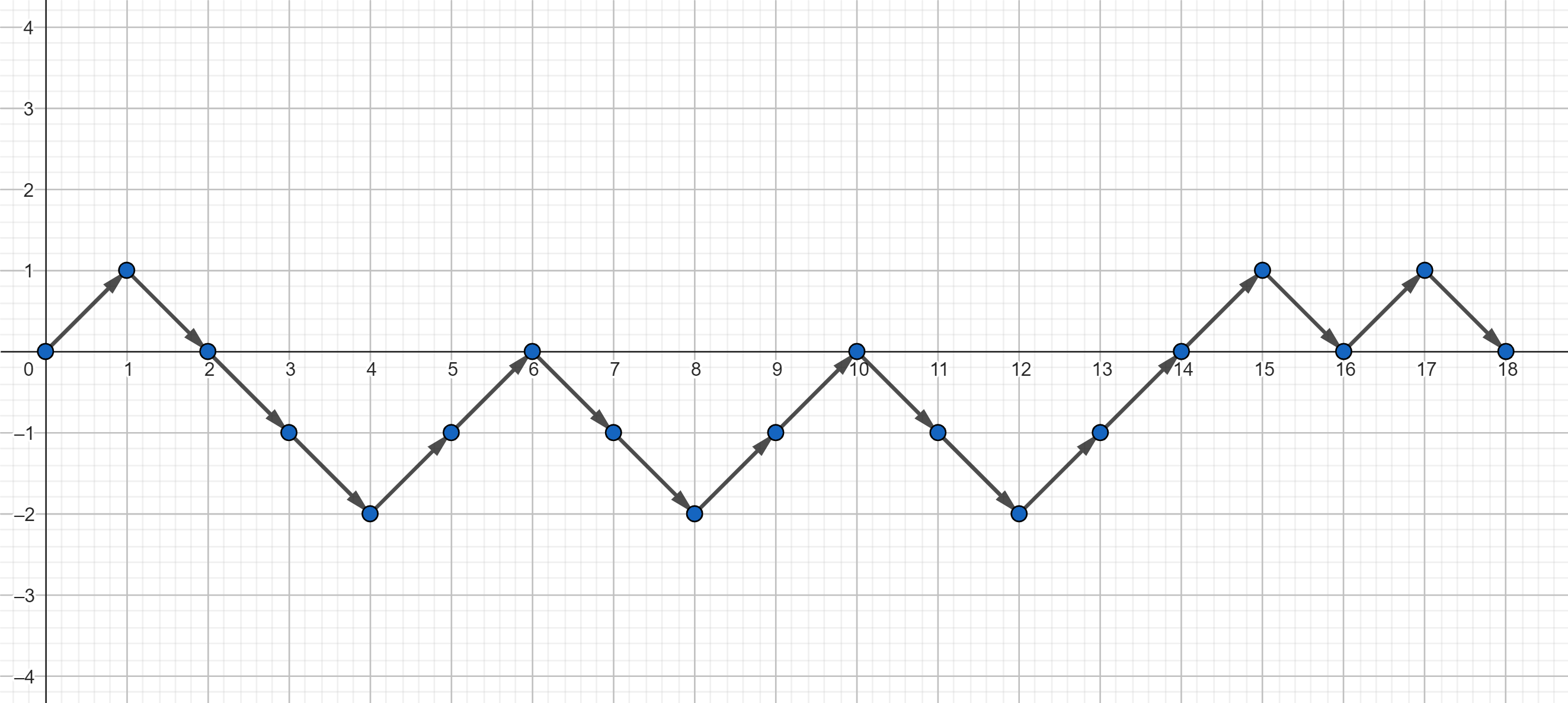}
    \caption{Illustration for the prefix $abbbaabbaabbaaabab$ of $f^{\omega}(a)$.}
    \label{fig:mpr1}
    \end{figure}

    \begin{figure}[h]
    \centering
    \includegraphics[width=1\linewidth]{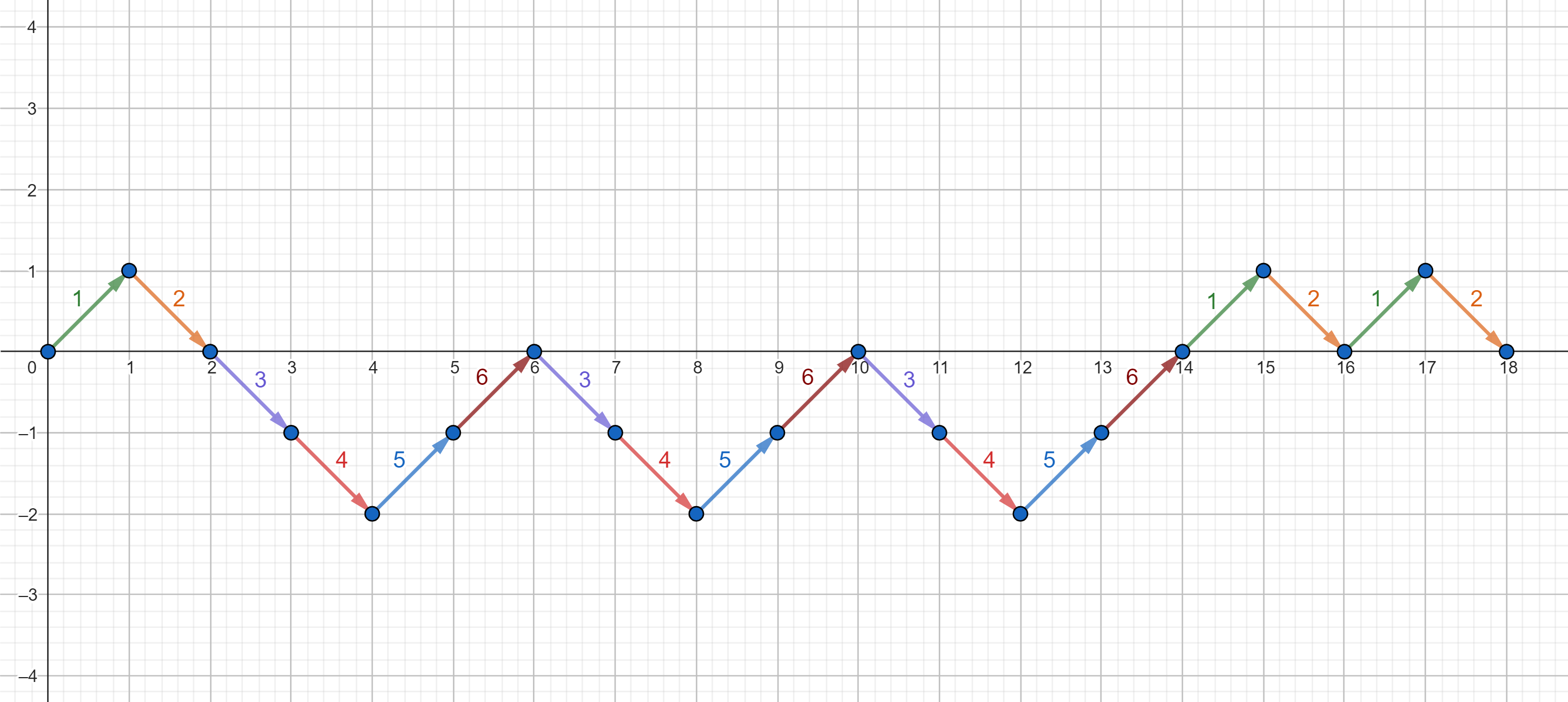}
    \caption{Arrows are marked with corresponding letters of the word $h^{\omega}(1)$}
    \label{fig:mpr2}
    \end{figure}




    We will make use of the following recent result:

    \begin{proposition}[Proposition 7 in \cite{ARI}]
    \label{prop:out}
    Let $g$ be a primitive uniform bijective morphism generating an aperiodic word. Then for each its fixed point $v$ the following holds: there is no infinite arithmetic progression such that at all positions from this progression $v$ has the same letter.
    \end{proposition}

    Using this result, we prove several auxiliary lemmas regarding properties of the morphism $h$.

    \begin{lemma}
        $\label{lemma:lem}$
        For each even $d\in \mathbb{N}$ there exist $k_1, k_2 \in \mathbb{N}$ such that at the positions $dk_1$ and $dk_2$ in the word $h^{\omega}(1)$ there are letters $3$ and $5$, respectively. In other words, the automaton from Fig. \ref{fig:mpr}, starting at state $1$ and reading ternary representations of $dk_1$ and $dk_2$, comes to states $3$ and $5$, respectively.
            
    \end{lemma}

    \begin{proof}
    The morhism $h$ satisfies the conditions of Proposition \ref{prop:out}, so the automaton leaves state $1$ reading ternary representation $[dk]_3$ of $dk$ for some $k \in \mathbb{N}$. Since $h$ is prolongable on the letter $3$, we can consider the word $h^{\omega}(3)$ generated by $h$. So, the automaton must also leave state $3$ reading $[dk]_3$ from state $3$ for some $k$. The same holds for state $5$. 
    
    We now analyze how reading $[dk]_3$ for some fixed $k \in \mathbb{N}$ can change the automaton states $1$, $3$ and $5$. Since $d$ is even, reading $[dk]_3$ can only put the automaton in one of these three states. Indeed, the state of the automaton after reading $[dk]_3$ corresponds to the letter at the position $dk$ in the word $h^{\omega}(1)$, and at even positions we can only have letters $1,\,3,\,5$. Besides that, the automaton is invertible, so $[dk]_3$ can do one of the following three things: either it does not change each of the three states $1$, $3$ and $5$, or it can trade places two of them keeping the third one unchanged, or it can move cyclically all the three states.

    As it is impossible that for each $k$ all ternary words $[dk]_3$ keep the state $1$ unchanged, there exists an integer $m_1 \in \mathbb{N}$, such that $[dm_1]_3$ changes the state $1$ to one of the states $3$ and $5$, and changes one of these two states to state $1$. If $[dm_1]_3$ changes the states cyclically, then at positions $[dm_1]_3$ and $[dm_1]_3[dm_1]_3$ (concatenation of ternary representation of $dm_1$ with itself) in the word $h^{\omega}(1)$ we have letters $3$ and $5$ (probably in different order) and, since $[dm_1]_3$ and $[dm_1]_3[dm_1]_3$ are ternary representations of numbers divisible by $d$, we have a proof of the lemma in this case.

    Now suppose that $[dm_1]_3$ trades places states $1$ and $i$ ($i$ = $3$ or $5$), and does not change the state $j = 8-i$. Then there exists $m_2 \in \mathbb{N}$ such that reading $[dm_2]_3$ changes state $j$ to some other state; more precisely, changes state  $j$  to one of the states $1$ or $i$. Then either reading $[dm_1]_3$ and $[dm_2]_3$ changes state $1$ to states $i$ and $j$, respectively, or reading $[dm_1]_3$ and $[dm_1]_3[dm_2]_3$ changes state $1$ to states $i$ and $j$, respectively. In either case we have a proof.    \end{proof}

    \begin{lemma}
        $\label{lemma:lemm}$
For each infinite arithmetic progression with even difference, the word $h^{\omega}(1)$ contains at least three  distinct letters at positions from this arithmetic progression. 
    \end{lemma}
        
    \begin{proof}     First note that since the difference of the arithmetic progression is even, all letters at corresponding positions in the word $h^{\omega}(1)$ have the same parity (this is easy to see from the definition of the morphism $h$).

    Assume the converse: suppose that such an arithmetic progression exists. Then its elements are of the form $r + dk,\; k\in \mathbb{N}_0$. Consider two cases depending on the parity of $r$.

    If $r$ is even, then possible letters at corresponding positions in the word $h^{\omega}(1)$ are $1,\;3$ and $5$. We will use automatic presentation of our word (see Fig. \ref{fig:mpr}). Suppose that ternary representation $[r]_3$ of the number $r$ changes state $1$ of our automaton to a state $a_1$, then $a_1 \in \{1,\;3,\;5\}$, since  $r$ is even. Similarly, suppose that $[r]_3$ changes state $3$ to a state $a_3$ and state $5$ to a state $a_5$; $a_3,\; a_5\in \{1,\;3,\;5\}$. Notice that the states $a_1,\; a_3 $ and $a_5$ are pairwise distinct, since our automaton is invertible. So,  reading $[r]_3$ from states from the set $\{1,\;3,\;5\}$, the automaton can come to each of the states $\{1,\;3,\;5\}$.

    By lemma $\ref{lemma:lem}$, the exist  $k_1,\;k_2 \in \mathbb{N}$, such that reading $[dk_1]_3$ and $[dk_2]_3$ from state $1$ the automaton comes to states $3$ and $5$, respectively. Then $[r]_3$, $[dk_1]_3 [r]_3$ (concatenation of ternary representation of $dk_1$ with ternary representation of $r$) and $[dk_2]_3 [r]_3$ move the automaton from state $1$ to each of the states $\{1,\;3,\;5\}$. It remains to notice that $[r]_3$, $[dk_1]_3 [r]_3$ and $[dk_2]_3 [r]_3$ are ternary representations of numbers of the form $r + dk$, $k \in \mathbb{N}_0$. Therefore, we found three elements of our arithmetic progression such that the correspoding positions in the word $h^{\omega}(1)$ contain all the three letters $\{1,\;3,\;5\}$.

    For even $r$ the proof is symmetric: suppose that reading $[r]_3$ changes the states $1$, $3$ and $5$ to the states $b_1$, $b_3$ and $b_5$, respectively. Since $r$ is odd, we have $b_1$, $b_3$, $b_5 \in \{2,\;4,\;6\}$, and in addition all $b_i$'s are disctinct, since the automaton is invertible. Thus we again found  three elements of the arthmetic progression such that the corresponding positions in the word $h^{\omega}(1)$ contain all three letters from $\{2,\;4,\;6\}$.
    \end{proof}

Now we are ready to prove Proposition \ref{pr:abel1}, i.e., we are going to prove that the word
 $f^{\omega}(a)$ is not abelian periodic. 

\begin{proof}[Proof of Proposition \ref{pr:abel1}]
We start with a simple observation: if the word is abelian periodic, then, since frequences of both letters $a$ and $b$ in $f^{\omega}(a)$ are equal to $\frac{1}{2}$, the abelian period must contain equal numbers of occurrences of letters $a$ and $b$ (in particular, it is even).

Consider the following geometric interpretation of the abelian periodicity property (note that this interpetation has been also used in \cite{WAP} for the study of a generalization of the notion of an abelian periodicity of infinite words). Let $w$ be a word on the alphabet $\{a,b\}$; we translate $w$ to a graphic visiting points of the infinite rectangular grid by interpreting letters of $w$ by drawing instructions. We
assign $a$ with a move by vector $v_a=(1, 1)$, and $b$ with a move by $v_b=(1, -1)$. We start at the origin $(x_0, y_0) = (0, 0)$. At step $n$, we are at a point $(x_{n-1}, y_{n-1})$ and
we move by the vector corresponding to the letter $w_n$, so that we come to the point
$(x_n, y_n) = (x_{n-1}, y_{n-1}) + v_{w_n}$.
So, we translate the word $w$ to a path in $\mathbb{Z}^2$ (see Fig. \ref{fig:mpr1}). 

Then, taking into account the observation we made about frequences of letters, the abelian periodicity of the word  $f^{\omega}(a)$ means that there exists a nonnegative integer $r$ and an even integer $d$, such that starting from a point on the graph of the word with $x=r$ (after $r$ steps from the origin), after each $d$ steps the graph has points with the same $y$-coordinate (since abelian period containt equal numbers of occurrences of $a$ and $b$). 

The following claim is straightforward:
    \begin{claim}
    \label{claim:lem_pos} 
    The graph of $f^{\omega}(a)$ satisfies the following:
    \begin{itemize}
        \item Points with $y=0$ correspond to prefixes consisting of several integral blocks of the word $f^{\omega}(a)$.
        \item Points with $y=1$ correspond to positions in the middle of the block $f(a) = ab$ (i.e., after $a$ in this block).
        \item Points with $y=-2$ correspond to positions in the middle of the block $f(b) = bbaa$ (i.e., after $bb$ in this block).
        \item Points with $y=-1$ correspond to one of the two positions in the block $f(b) = bbaa$: either after the first $b$, or after the first $a$.
        \end{itemize}
    \end{claim}
    

    Claim \ref{claim:lem_pos} together with the proof of Lemma \ref{lemma:l1} give a relation between the graph of the word $f^{\omega}(a)$ and the word $h^{\omega}(1)$ (see Fig. \ref{fig:mpr2}). On the figure, each arrow is marked by the corresponding letter of the word $h^{\omega}(1)$.
    We associate points of the graph of $f^{\omega}(a)$ with letters of $h^{\omega}(1)$ as follows: for each integer $m$, the point corresponding to the prefix of length $m$ of $f^{\omega}(a)$ is associated with the letter $h^{\omega}(1)_m$.
    On Fig. \ref{fig:mpr2} letters associated with points correspond to letters marking arrows going from this point.
    We can now reformilate Claim \ref{claim:lem_pos} in terms of the word $h^{\omega}(1)$.

    \begin{claim}
    \label{claim:lem_pos2}
    Letters of $h^{\omega}(1)$ assoiciated with points of the graph of $f^{\omega}(a)$ satisly the following:
    \begin{itemize}
        \item Points with $y=0$ are associated with letters $1$ and $3$ 
        \item Points with $y=1$ are associated with letter $2$.
        \item Points with $y=-1$ are associated with $4$ and $6$.
        \item  Points with $y-2$ are associated with letter $5$.
    \end{itemize}
    \end{claim}

    By Lemma \ref{lemma:lemm}, we obtain that there is no aritmetic progression with even difference $d$, such that corresponding positions in the word  $h^{\omega}(1)$ are filled with at most two distinct letters. This completes the proof: Indeed, in case of abelian periodicity with period $d$ (which must be even) and preperiod $r$, all points corresponding to prefixes of length $r + dk,\; k\in \mathbb{N}$, of $f^{\omega}(a)$, are on the same horizontal axis. So, by Claim \ref{claim:lem_pos2}, positions $r + dk$ in $h^{\omega}(1)$ contain at most two distinct letters.
\end{proof}

\subsection{General case with $\theta_2=0$}
\label{subsec:parikh_coll}


In this subsection, we consider binary morphisms $f$ with matrix $M_f$ having the second eigenvalue equal to 0, i.e., $\theta_2 = 0$. We can then represent $M_f$ in the following form described in Condition 1: 

\medskip

\textbf{Condition 1}:  The matrix of $f$ is of the following form: $$M_f = \begin{pmatrix}
n A & m A\\
n B & m B
\end{pmatrix}$$ with $A,B,m,n \in \mathbb{N}$ and $(m, n) = 1$. 

\medskip

Besides that, we suppose that $f$ satisfies the following:

\medskip

\textbf{Condition 2}: $f$ is a binary morphism of the form $$\begin{cases}
    f(a) = ax, \\
    f(b) = by,
\end{cases}$$ where $x$ and $y$ are non-empty words.

\medskip

We will state and prove auxiliary statements for morphisms saisfying Condition 2, and then in the proof of Theorem \ref{CONJ}, we will show that it is enough to consider morphisms of this form. We now introduce the following definitions for a morphism $f$ satisfying Condition 1:



\begin{definition}
     Let and let $u$ be a factor of the word $f^{\omega}(a)$ occurring at position $i$, i.e. $f^{\omega}(a)[i..i+|u|-1]=u$. A \textit{block-position} of this occurrence of $u$ in $f^{\omega}(a)$ is the number  $\frac{i}{A+B}$ if it is integer (otherwise a block-position is not defined).  
\end{definition}

Note that we can similarly define a block-position in a finite word of the form $f^n(a)$.

\begin{definition}
    Let $t\in \mathbb{N}$, and let $u$ be a factor of $f^{\omega}(a)$ occurring at position $i$. Then a \textit{$t$-block-position} of this occurrence of $u$ in  $f^{\omega}(a)$ is the number $ \frac{i}{(A + B)(nA + mB)^{t-1}}$  if it is integer (otherwise a $t$-block-position is not defined). 
\end{definition}

\begin{definition} \label{def_BL} Given a word $u$, a \textit{block-length} of the word $f(u)$ is the number $$BL_{f(u)} = \frac{|f(u)|}{A + B}.$$
\end{definition}

\begin{definition}
We say that an occurrence of a word $f(w)$ in $f^{\omega}(a)$ is \textit{proper} if we have for this occurrence  $f^{\omega}(a) = f(u)f(w)\cdots$ for some prefix $u$ of $f^{\omega}(a).$ 
\end{definition}

\subsubsection{On block-positions of factors in $f^{\omega}(a)$}

In this subsection we discuss properties of block-positions of factors in $f^{\omega}(a)$ for binary morphisms satisfying Conditions 1 and 2. The main result of this subsection is Proposition \ref{tBPos}, which is the key tool in the proof of Theorem \ref{th:main}.

We start with a basic number theoretcal observation which we will use several times in the proofs. Although this is likely to be known or well understood, we provide a proof for completeness.

\begin{lemma}
\label{coprime}
Let $N$ and $d$ be relatively prime integers. Then for each $r\in \mathbb{N}$ the following holds: $$\{\alpha rN \bmod  d  \mid \alpha\in \{0,\,1,\dots,\,d-1\} \} = \{ \alpha r  \bmod d \mid \alpha\in \{0,\,1,\dots,\,d-1\}\}.$$
\end{lemma}

\begin{proof}
    Since $N$ and $d$ are relatively prime, by Bezout's identity, there exist integers  $x$ and $y$ such that $xN + yd = 1$. So, $xN \equiv_d 1$, and therefore 
    $$\begin{aligned}
    \{\alpha N \bmod  d  \mid \alpha\in \{0,\,1,\dots,\,d-1\} \} \supset \{\alpha x N \bmod d  \mid \alpha\in \{0,\,1,\dots,\,d-1\} \} = \\ = \{\alpha \bmod d  \mid \alpha\in \{0,\,1,\dots,\,d-1\} \} = \{0,\,1,\dots,\,d-1\}.
    \end{aligned}$$
    Hence $$\{\alpha N \bmod d  \mid \alpha\in \{0,\,1,\dots,\,d-1\} \} = \{0,\,1,\dots,\,d-1\}.$$ Therefore, $$\begin{aligned}
    \{\alpha rN \bmod d  \mid \alpha\in \{0,\,1,\dots,\,d-1\} \} = \{(\alpha N) r \bmod d  \mid \alpha\in \{0,\,1,\dots,\,d-1\} \} = \\ = \{ \alpha r  \bmod d \mid \alpha\in \{0,\,1,\dots,\,d-1\}\}. 
    \end{aligned}  $$
\end{proof}

\begin{lemma}
\label{laa}
Let $f$ be a binary morphism satisfying Conditions 1 and 2.
Assume in addition that there is an occurrence of $aa$ in the word $f^{\omega}(a)$. Then for each an integer  $d$  coprime with $\text{tr}(M_f)$ 
and for each $\alpha \in \{0,\,1,\dots,\,d-1\}$, there is a proper occurrence of $f(a)$ in $f^{\omega}(a)$ at a block-position congruent to $\alpha n$ modulo $d$.
\end{lemma}

\begin{proof} Since $f^{\omega}(a)$ contains an occurrence of $aa$, we can assume that there is an occurrence of $aa$ in the block $f(a)$. Otherwise, consider a morphism $f^k$ for sufficiently large $k$ instead of $f$; we only need to check that if the statement holds for $f^k$, then it holds for $f$. 
We have $\text{tr}(M_{f^k}) = (nA + mB)^k$; suppose that for each $\alpha \in \{0,\,1,\dots,\,d-1\}$ there is an occurrence of $f^k(a)$ at a block-position (for $f^k$) comparable to $\alpha n$ modulo $d$. For each such occurrence of $f^k(a)$, its prefix $f(a)$ has a block-position (for $f$) comparable to $(nA + mB)^{k-1}\alpha n $, since the length of a block of the morphism $f^k$ is equal to the length of a block of $f$ multiplied by $(nA + mB)^{k-1}$. Now, since $(nA + mB)^{k-1}$ is relatively prime with $d$, by Lemma \ref{coprime}, we have that $$\{ (nA + mB)^{k-1}\alpha n \bmod d  \mid \alpha\in \{0,\,1,\dots,\,d-1\} \} = \{ \alpha n  \bmod d \mid \alpha\in \{0,\,1,\dots,\,d-1\}\}.$$
 
 
 So, it is enough to prove the statement for $f^k$, and hence we can assume that  $f(a)$ contains an occurrence of $aa$.
Thus, $f(a) = uaav$, where $u,\,v\in \{a,\,b\}^*$. 

\begin{claim}
\label{pr1}
   For each sequence of positive integers $k_1 > k_2 > k_3 > \ldots > k_t$, the word
   $$f^{k_1}(u)f^{k_2}(u)f^{k_3}(u)\cdots f^{k_t}(u)f^{k_t}(a)$$
   is a prefix of $f^{k_1 + 1}(a)$ and hence of $f^{\omega}(a)$.
\end{claim}

\begin{proof}
    We prove it by induction on $t$. Base case $t = 1$ is obvious: $f^{\omega}(a)$ starts with $f^{k+1}(a)$ for each integer $k$, and $f^{k+1}(a)=f^{k}(uaav)=f^{k}(u)\cdots$. Now consider $t+1$ numbers $k_1 > k_2 > \dots > k_{t+1}$. By the induction hypothesis, $f^{k_1 + 1}(a) = f^{k_1}(u)f^{k_2}(u)f^{k_3}(u)\cdots f^{k_t}(u)f^{k_t}(a)\cdots$, and also, $f^{k_t}(a) = f^{k_{t+1} + 1}(a)\cdots = f^{k_{t+1}}(u)f^{k_{t+1}}(a)\cdots$, and, combining these, we obtain the required result: 
   
  $$ \begin{aligned}
  f^{k_1 + 1}(a) = f^{k_1}(u)f^{k_2}(u)f^{k_3}(u)\cdots f^{k_t}(u)f^{k_t}(a)\cdots =  \\
  = f^{k_1}(u)f^{k_2}(u)f^{k_3}(u)\cdots f^{k_t}(u)f^{k_{t+1}}(u)f^{k_{t+1}}(a)\cdots
    \end{aligned}
  $$
\end{proof}

\begin{claim}
\label{pr2}

   For each sequence of positive integers $k_1 > k_2 > k_3 > \ldots > k_t$, the word $$f^{k_1}(ua)f^{k_2}(ua)f^{k_3}(ua)\cdots f^{k_t}(ua)f^{k_t}(a)$$ is a prefix of  $ f^{k_1 + 1}(a)$ and hence of $f^{\omega}(a)$.
\end{claim}

\begin{proof}
    The proof is by induction on $t$. 
    Base case $t = 1$ is straightforward: $f^{k+1}(a)=f^k(uaav)=f^k(ua)\cdots$ for each integer $k$. 
    Now consider $t+1$ numbers $k_1 > k_2 > \ldots > k_{t+1}$. By the induction hypothesis, we have $$f^{k_1 + 1}(a) = f^{k_1}(ua)f^{k_2}(ua)f^{k_3}(ua)\cdots f^{k_t}(ua)f^{k_t}(a)\cdots$$ and $$f^{k_t}(a) = f^{k_{t+1} + 1}(a)\dots = f^{k_{t+1}}(ua)f^{k_{t+1}}(a)\cdots$$ Combining these, we obtain the required result: 
   
  $$ \begin{aligned}
  f^{k_1 + 1}(a) = f^{k_1}(ua)f^{k_2}(ua)f^{k_3}(ua)\cdots f^{k_t}(ua)f^{k_t}(a)\cdots =  \\
  = f^{k_1}(ua)f^{k_2}(ua)f^{k_3}(ua)\cdots f^{k_t}(ua)f^{k_{t+1}}(ua)f^{k_{t+1}}(a)\cdots 
      \end{aligned}
  $$
\end{proof}

Now let $\alpha \in \{1,\,2,\,\dots,\,d-1\}$. Consider the sequence $$k_i = (\alpha -i+1)\varphi(d) + 1 \mbox{ for } i = 1,\, 2,\, \dots,\, \alpha,$$ where $\varphi$ is Euler's totient function. By Claim \ref{pr2}, we obtain:
 \begin{gather}
f^{\omega}(a) =
f^{\alpha \varphi(d) + 2}(a)\cdots = \nonumber \\ = f^{\alpha \varphi(d) + 1}(ua)f^{(\alpha - 1) \varphi(d) + 1}(ua)f^{(\alpha - 2) \varphi(d) + 1}(ua)\cdots f^{\varphi(d) + 1}(ua)f^{\phi(d) + 1}(a)\cdots = \label{eq:BL} \\ = f^{\alpha \varphi(d) + 1}(ua)f^{(\alpha - 1) \varphi(d) + 1}(ua)f^{(\alpha - 2) \varphi(d) + 1}(ua)\cdots f^{\varphi(d) + 1}(ua)\mathbf{f(a)}\cdots. \nonumber
  \end{gather}
For each word $w$ we have $|f^2(w)|=|f(w)|(nA+mB)$, so
for each $k\in \mathbb{N}$ 
we have 
\begin{equation}\label{eq:BL1}\frac{|f^{k\varphi(d) + 1}(w)|}{A+B} = (nA + mB)^{k\varphi(d)}\frac{|f(w)|}{A+B}.\end{equation} 
By Euler's theorem we have $c^{\varphi(m)}\equiv_m 1$, if $c$ and $m$ are relatively prime. By the conditions of the lemma we have that  $(nA + mB)$ and $d$ are relaively prime, hence $(nA + mB)^{k}$ and $d$ are relatively prime as well, so 
\begin{equation}\label{eq:BL2} (nA + mB)^{k\varphi(d)}\frac{|f(w)|}{A+B} \equiv_d \frac{|f(w)|}{A+B}.\end{equation}
Considering the occurrence of $f(a)$ in \eqref{eq:BL} marked by bold and applying \eqref{eq:BL1} and \eqref{eq:BL2} to it, we obtain an occurrence of $f(a)$ in $f^{\omega}(a)$ with a block-position congruent to $$\alpha (BL_{f(u)} + BL_{f(a)}) = \alpha BL_{f(u)} + \alpha n$$  modulo $d$ (we recall Definition \ref{def_BL} for notation $BL$). 


Now, consider $\beta = d-\alpha$. Similarly to the construction above and applying Claim $\ref{pr1}$, we get 
 \begin{gather*}
f^{\omega}(a) =
f^{\beta \varphi(d) + 2}(a)\cdots = \\ = f^{\beta \varphi(d) + 1}(u)f^{(\beta - 1) \varphi(d) + 1}(u)f^{(\beta - 2) \varphi(d) + 1}(u)\cdots f^{\varphi(d) + 1}(u)\mathbf{f(a)}\cdots.
 \end{gather*}
This implies that we have an occurrence of $f(a)$ in $f^{\omega}(a)$ with a block-position congruent to $\beta BL_{f(u)}$ modulo $d$.

Finally, consider $f^{k\varphi(d)}(f^{\alpha \varphi(d) + 2}(a))$ for some $k \in \mathbb{N}$ such that $k\varphi(d) \geq \beta\phi(d) + 2$:
\begin{gather*}
    f^{k\varphi(d)}(f^{\alpha \varphi(d) + 2}(a)) = \\ = f^{k\varphi(d)}\Bigl(f^{\alpha \varphi(d) + 1}(ua)f^{(\alpha - 1) \varphi(d) + 1}(ua)f^{(\alpha - 2) \varphi(d) + 1}(ua)\cdots f^{\varphi(d) + 1}(ua)\mathbf{f(a)}\cdots\Bigr).
\end{gather*}
Hence, we have an occurrence of $f^{\beta  \varphi(d) + 2}(a)$ with a block-position congruent to $\alpha BL_{f(u)} + \alpha n$ modulo $d$. Then, within this occurrence of $f^{\beta \varphi(d) + 2}(a)$ we find an occurrence of $f(a)$ with a block-position in the word $f^{\beta \varphi(d) + 2}(a)$
congruent to $\beta BL_{f(u)}$ modulo $d$. Combining it with a vlock-position of $f^{\beta \varphi(d) + 2}(a)$ in $f^{\omega}(a)$, we find an occurrence of $f(a)$ with a block-position in the word $f^{\omega}(a)$ that is congruent to $$(\alpha + \beta) BL_{f(u)} + \alpha n \equiv_d \alpha n$$ modulo $d$.

It remains to consider the case $\alpha = 0$: The prefix occurrence of $f(a)$ in $f^{\omega}(a)$ has a block-position $0$.

This completes the proof of Lemma \ref{laa}.
\end{proof}

\begin{lemma}
\label{lbb}
Let $f$ be a binary morphism satisfying Conditions 1 and 2.
Assume in addition that there is an occurrence of $bb$ in the word $f^{\omega}(a)$. Then for each integer $d$  coprime with $\text{tr}(M_f)$, 
 there is a constant $C$ such that for each $\alpha \in \{0,\,1,\dots,\,d-1\}$ there is a proper occurrence of $f(b)$ in $f^{\omega}(a)$ at a block-position congruent to $\alpha m+C$ modulo $d$.
\end{lemma}

\begin{proof}
Consider the other fixed point $f^{\omega}(b)$ of the same morhism. Applying Lemma \ref{laa} to it, we get proper occurrences of $f(b)$ in all positions $\alpha m$. Now we take sufficiently long prefix $f^k(b)$; it has a proper occurrence in $f^{\omega}(a)$. Its block-position gives the constant $C$.
\end{proof}

\begin{lemma}
\label{laabb}
Let $f$ be a binary morphism satisfying Conditions 1 and 2.
Assume in addition that there are occurrences of both $aa$ and $bb$ in the word $f^{\omega}(a)$. Then for each integer $d$ coprime with $\text{tr}(M_f)$ 
and for each $\alpha \in \{0,\,1,\dots,\,d-1\}$, block-positions of proper occurrences of $f(a)$ in $f^{\omega}(a)$ run through all residues modulo $d$.
\end{lemma}

\begin{proof}
  First, note that numbers of the form $$\alpha n + \beta m$$ with $\alpha,\,\beta \in \{0,\,1,\dots,\,d-1\}$ run through all residues modulo $d$. Indeed, since $d$ is relatively prime with $\text{tr}(M_f)=nA + mB$, by Lemma \ref{coprime} we have that the numbers $nA+mB$, $2(nA+mB)$, $3(nA+mB)$, \dots, $(d-1)(nA+mB)$ give all $d$ different residudes modulo $d$.  


    Secondly, for each $k\in \mathbb{N}$ the factor $f^k(a)$ in $f^{\omega}(a)$ occurs at block-positions congruent to $\alpha n$ modulo $d$ for each $\alpha \in \{0,\,1,\dots,\,d-1\}$. Indeed, by Lemma $\ref{laa}$, this holds for $k = 1$; then the factors $f^k(a)$ in $f^{\omega}(a)$ occur at block-positions congruent to $$\alpha n (nA + mB)^{k-1}$$ modulo $d$ for each $\alpha \in \{0,\,1,\dots,\,d-1\}$. Since $d$ and $nA + mB$ are relatively prime, we conclude.
    
    Consider $k\in \mathbb{N}$ such that for some $C$ the word $f^k(a)$  contains proper occurrences of $f(b)$ with block-positions congruent to $\beta m + C$ modulo $d$, for each $\beta \in \{0,\,1,\dots,\,d-1\}$ (by Lemma $\ref{lbb}$, such $k$ exists). Then in $f^{k+1}(a)$ there are proper occurrences of the word $f^2(b)$ with block-positions congruent to $(\beta m + C)(nA+mB)$ modulo $d$, for each $\beta \in \{0,\,1,\dots,\,d-1\}$, that is, by Lemma \ref{coprime}, block-positions congruent to $\beta m + C'$ modulo $d$, for each $\beta \in \{0,\,1,\dots,\,d-1\}$ (the constant $C'$ does not depend on $\beta$). Now fix $\alpha,\,\beta \in \{0,\,1,\dots,\,d-1\}$. Consider an occurrence of $f^{k+1}(a)$ in $f^{\omega}(a)$ with a block-position congruent to $\alpha n$ modulo $b$. Within this occurrence of $f^{k+1}(a)$, consider an occurrence of $f^2(b)$ with block-position (in $f^{k+1}(a)$) congruent to $\beta m + C'$ modulo $d$; block-position of this occurrence in the word $f^{\omega}(a)$ is congruent to $\alpha n + \beta m + C'$ modulo $d$. Within the found occurrence of $f^2(b)$, consider the first occurrence of $f(a)$; its block-position in $f^{\omega}(a)$ is congruent to $\alpha n + \beta m + C''$ modulo $d$ (the constant $C''$ does not depend on $\alpha$ and $\beta$). Note that the residues modulo $d$ of numbers of the form $\alpha n + \beta m + C''$ run through all possible values modulo $d$ if $\alpha$ and $\beta$ run through  $\{1,\,2,\,\dots,\,d-1\}$. 
    Lemma \ref{laabb} is proved.
\end{proof}

\begin{lemma}
\label{laab}
Let $f$ be a binary morphism satisfying Conditions 1 and 2. 
Assume in addition that there is an occurrence of $aa$ and there is no occurrence of $bb$ in the word $f^{\omega}(a)$. Then for each integer  $d$  coprime with $\text{tr}(M_f)$  
and for each $\alpha \in \{0,\,1,\dots,\,d-1\}$, there is a proper occurrence of $f(a)$ in $f^{\omega}(a)$ at a block-position congruent to $\alpha n$ modulo $d$.
\end{lemma}

\begin{proof}
Since the word $f^{\omega}(a)$ does not contain occurrences of $bb$, each occurrence of $b$ in $f^{\omega}(a)$ is followed by an occurrence of $a$. Consider sufficiently large $k$, such that the word $f^k(a)$ contains at least $d$ occurrences of $b$. Then for each $\beta \in \{0,\,1,\dots,\,d-1\}$ the word $f^{k+1}(a)$ contains an occurrence of the word $f(a)$ with a block-position equal to $\gamma_{\beta} n + \beta m$ for some $\gamma_{\beta}$ depending on $\beta$. Indeed, for each proper occurrence of $f(b)$ (resp., $f(a)$) with a block position $t_b$ (resp., $t_a$) the block-position of a proper occurrence of a factor following it is $t_b+m$ (resp., $t_a+n$). Since $f^k(a)$ contains at least $d$ occurrences of $b$, then $f^{k+1}(a)$ contains at least $d$ proper occurrences of $f(b)$. Since we do not have occurrences of $bb$ in $f^\omega(a)$, each such proper occurrence of $f(b)$ is followed by a proper occurrence of $f(a)$. Thus the prefix occurrence of $f(a)$ corresponds to  $\beta = 0$, after that we have several occurrences of $f(a)$ (say, $\gamma_1$ occurrences), then we have the first block $f(b)$, after that an occurrence of a block $f(a)$ corresponding to $\beta = 1$, and so on.


Fix some $\alpha,\,\beta \in \{0,\,1,\dots,\,d-1\}$.
Consider an occurrence of the word $f(a)$ in the word $f^{k+1}(a)$ at a block-position congruent to $\gamma_{\beta} n + \beta m$ modulo $d$. Now, similarly to the proof of Lemma $\ref{laabb}$, we find an occurrence of the word $f^{k+1}(a)$ in the word $f^{\omega}(a)$ with a block-position congruent to $(\alpha - \gamma_{\beta})n$ modulo $d$. Then the block-position of the occurrence of $f(a)$ in $f^{\omega}(a)$ is congruent to $\alpha n + \beta m$ modulo $d$.

Lemma \ref{laab} is proved.
\end{proof}

\begin{lemma}
\label{LAA}
 Let $f$ be a binary morphism satisfying Conditions 1 and 2. 
Suppose in addition that $f^{\omega}(a)$ contains occurrences of $aa$. Then for each $d\in \mathbb{N}$ coprime with $\text{tr}(M_f)$, 
 block-positions of proper occurrenсes of $f(a)$ in $f^{\omega}(a)$ run through all residues modulo $d$.
\end{lemma}

\begin{proof}
    Combination of Lemma $\ref{laabb}$ and Lemma $\ref{laab}$.
\end{proof}

\begin{lemma}
\label{LBB}
Let $f$ be a binary morphism satisfying Conditions 1 and 2. 
Suppose in addition that $f^{\omega}(a)$ contains occurrences of $bb$. Then for each $d\in \mathbb{N}$ coprime with $\text{tr}(M_f)$, 
 block-positions of proper occurrenсes of $f(a)$ in $f^{\omega}(a)$ run through all residues modulo $d$.
\end{lemma}

\begin{proof}
    By Lemma \ref{LAA} applied to the word $f^{\omega}(b)$ we have that block-positions of proper occurrenсes of $f(b)$ in $f^{\omega}(b)$ run through all residues modulo $d$. Consider $k \in \mathbb{N}$, such that the prefix  $f^k(b)$ contains proper occurrences of $f(b)$ at block-positions running through all residues modulo $d$. Then block-positions of proper occurrences of $f^2(b)$ in $f^{k+1}(b)$ also run through all residues modulo $d$. Indeed, $nA+mB$ is relatively prime with $d$, and for a proper occurrence of $f(b)$ in $f^k(b)$ with a block-position $i$ we have a proper occurrence of $f^2(b)$ in $f^{k+1}(b)$ with a block-position $i(nA+mB)$. So, by Lemma \ref{coprime}, we have $\{i(nA+mB) \mid i\in\{0, 1\dots, d-1\}\} = \{0, 1\dots, d-1\}$. 
    Now, block-positions of proper occurrences of $f(a)$ in $f^{k+1}(b)$ also run through all residues modulo $d$: we can consider for example the first proper occurrences of $f(a)$ in $f^2(b)$.
     We can express $f^{\omega}(a)$ in the form $f^{\omega}(a) = f^{k+1}(w)f^{k+1}(b)\cdots$ for some prefix $w$ of  $f^{\omega}(a)$, and in this occurrence of $f^{k+1}(b)$ we can consider all proper occurrences of $f(a)$ at block-positions running through all residues modulo $d$. Thus block-positions of these occurrences of $f(a)$ in $f^{\omega}(a)$ also run through all residues modulo $d$.
\end{proof}

\begin{lemma} 
\label{l1} 
Let $f$ be a binary morphism satisfying Conditions 1 and 2. Then for each $d\in \mathbb{N}$ coprime with $\text{tr}(M_f)$, 
 block-positions of proper occurrenсes of $f(a)$ in $f^{\omega}(a)$ run through all residues modulo $d$.
\end{lemma}

\begin{proof}
    If the word $f^{\omega}(a)$ has an occurrence of $aa$, we use Lemma \ref{LAA}; if it has an occurrence of $bb$, we use Lemma \ref{LBB}. Suppose that $f^{\omega}(a)$ does not have occurrences of $aa$ and $bb$. Then $f^{\omega}(a) = (ab)^{\omega}$. It is easy to see that each morphism
     $f$ satisfying Condition 1 and generating the word $(ab)^{\omega}$ must be of the form $$\begin{cases}
        f(a) = (ab)^k \\
        f(b) = (ab)^t.
    \end{cases}$$ The matrix of this morphism is $M_f = \begin{pmatrix}
k & t\\
k & t
\end{pmatrix}$. By the conditions of the lemma we have that $d$ is relatively prime with $\text{tr}(M_f)=k+t$. For each integer $r$ we can express $f^{\omega}(a)$ as $$f^{\omega}(a) = (f(a)f(b))^{\omega} = (f(a)f(b))^rf(a)\cdots$$ The block-lengths of the words $(f(a)f(b))^r =(ab)^{r(k+t)}$  
run through all residues modulo $d$ when $r = 0,\,1,\,\dots,\,d-1$.
Lemma is proved.    
\end{proof}

We are now ready to prove the main proposition of this subsection:

\begin{proposition}
\label{tBPos}
Let $f$ be a binary morphism satisfying Conditions 1 and 2. Given $d\in \mathbb{N}$ coprime with $\text{tr}(M_f)$, 
suppose that for some $t\in \mathbb{N}$ a factor $w$ of the word $f^{\omega}(a)$ has an occurrence at a position divisible by  $(A+B)(nA + mB)^{t-1}$ (i.e., for at least one occurrence of $w$ a $t$-block-position is defined). Then $t$-block-positions of occurrences of $w$ in $f^{\omega}(a)$ run through all residues modulo $d$. 
\end{proposition}

\begin{proof}

We first note that it is enough to prove the statement for words of the form  $f^{k}(a)$, since each factor $w$ of $f^{\omega}(a)$ is a factor of $f^{k}(a)$ for some $k$. So, if $t$-block-positions of occurrences of $f^{k}(a)$ run through all residues modulo $d$, then this also holds for $t$-block-positions of occurrence of $w$.

By Lemma $\ref{l1}$, block-positions of occurrences of the word $f(a)$ run through all residues modulo $d$. Consider $k\in \mathbb{N}$. If there is a proper occurrence of $f(a)$ at a block-position $i$ as $f^{\omega}(a) = f(u)\mathbf{f(a)}\cdots$, then we have an occurrence of $f^{k+t-1}(a)$ as $f^{\omega}(a) = f^{k+t-1}(u)\mathbf{f^{k+t-1}(a)}\cdots$, i.e., at a block-position $i(nA + mB)^{k+t-2}$, and at a $t$-block-position $i(nA + mB)^{k-1}$. Since the value $i$ runs through all residues modulo  $d$, then, by Lemma \ref{coprime}, the value $i(nA + mB)^{k-1}$ does as well, as $nA + mB$ is relatively prime with $d$. The word $f^k(a)$ is a prefix of $f^{k+t-1}(a)$, so $t$-block-positions of $f^{k}(a)$ run through all residues modulo $d$.
\end{proof}



\subsubsection{Necessary and sufficient condition for abelian periodicity in the case $\theta_2 = 0$}


The following is a key lemma for the proof of Theorem \ref{th:main} in the case $\theta_2 = 0$:

\begin{lemma}
\label{lemma:rem_d2}
Let $f$ be a binary morphism satisfying Conditions 1 and 2. Suppose that $f^{\omega}(a)$ is an  ultimately abelian periodic word with abelian period $s = d(A + B)(nA + mB)^{t-1},$ where $t$ and $d$ are integers such that $d$ is relatively prime with $(nA + mB)$. Consider a factorization of $f^{\omega}(a)$ of the form $$f^{\omega}(a) = v\cdot\prod\limits_{i = 0}^{\infty} u_i,$$ where $|u_i| = (A + B)(nA + mB)^{t-1}$, for each $i$, and where $v$ is a preperiod. Then for each $j \in \mathbb{N}$ the word $u_j$ is abelian equivalent to $u_0$.
\end{lemma}

In other words, the lemma says that if a word generated by a morphism satisfying Conditions 1 and 2 has an abelian period divisible by $(A + B)(nA + mB)^{t-1}$, then the number $(A + B)(nA + mB)^{t-1}$ is also an abelian period.


\begin{proof}
    First, consider $j = kd$ for some $k\in \mathbb{N}$. Assume that $u_{kd}$ is not abelian equivalent to $u_0$. Let $k_0$ be the minimal $k$ such that $u_{kd}$ is not abelian equivalent to $u_0$. We let $w$ denote the following prefix of $f^{\omega}(a)$: $$w = v\cdot\prod\limits_{i = 0}^{dk_0} u_i.$$ For the prefix occurrence of $w$ its $t$-block-position is defined, so, by Lemma \ref{tBPos}, there exists an occurrence of $w$ in $f^{\omega}(a)$ at a $t$-block-position congruent to $d-1$ modulo $d$, i.e., $$w = v\cdot\prod\limits_{i = dl - 1}^{dl - 1 + dk_0} u_i$$ for some $l\in \mathbb{N}$. As $d(A + B)(nA + mB)^{t-1}$ is an abelian period of $f^{\omega}(a)$, the words $\prod\limits_{i = 0}^{dk_0 - 1} u_i$ and $\prod\limits_{i = dl}^{dl - 1 + dk_0} u_i$ are abelian equivalent. This implies that $u_{dk_0} \sim_{ab} u_{dl-1}$. We have $u_{dl-1} = u_0$, as these words are factors of $w$ at the same positions. Thus $u_{dk_0} \sim_{ab} u_0$. A contradiction.

    Now, consider an arbitrary $j\in \mathbb{N}$. We prove that $u_j \sim_{ab} u_0$. Let us denote $$w = v\cdot\prod\limits_{i = 0}^{j} u_i.$$ The word $w$ is a prefix of $f^{\omega}(a)$. For a prefix occurence of $w$ in $f^{\omega}(a)$, its $t$-block-position is well-defined, so, by Lemma \ref{tBPos}, there exists an occurrence of $w$ in $f^{\omega}(a)$ at a $t$-block-position congruent to $-j$ modulo $d$, i.e. $$w = v\cdot\prod\limits_{i = dl - j}^{dl - j + j} u_i = v\cdot\prod\limits_{i = dl - j}^{dl} u_i$$ for some $l\in \mathbb{N}$. That gives us $u_j = u_{dl}$, as these are both suffixes of the word $w$ of the same length. Thus, as $u_{dl} \sim_{ab} u_0$, we obtain $u_j \sim_{ab} u_0$.
\end{proof}

The following is a base for the proof of  Theorem \ref{th:main} in the case $\theta_2=0$:

\begin{theorem}
\label{CONJ}
    Let $f$ be a binary morphism with matrix 
$M_f = \begin{pmatrix}
n A & m A\\
n B & m B
\end{pmatrix}$, where $m,\,n \in \mathbb{N}$ and $(m,\, n) = 1$. Suppose that $f^{\omega}(a)$ is eventually abelian periodic. Then it has an abelian period $$\text{gcd}(|f^K(a)|,\;|f^K(b)|) = (A + B)(nA + mB)^{K-1}$$ for some $K \in \mathbb{N}$. 
\end{theorem}

\begin{proof}
First note that we can assume without loss of generality that $f$ is of the form \begin{equation}\label{eq_f_ab}\begin{cases}
    f(a) = ax \\
    f(b) = by.
\end{cases}\end{equation} Indeed, if the blocks $f(a)$ and $f(b)$ start with the same letter, then we can consider a cyclic shift and repeat the process until the first letters become distinct. If this does not happen, then the blocks $f(a)$ and $f(b)$ are powers of the same word; in this case the word $f^{\omega}(a)$ is obviously periodic and the statement of the theorem holds. If after some cyclic shift we have a morphism of the form  $$\begin{cases}
    \widetilde{f}(a) = b\widetilde{x} \\
    \widetilde{f}(b) = a\widetilde{y},
\end{cases}$$ consider its square $\widetilde{f}^2$. 

We will now prove that the words $(\widetilde{f}^2)^{\omega}(a)$ and $f^{\omega}(a)$ are either both abelian periodic or both not abelian periodic, and moreover the abelian period is the same if exists.

Let $u$ be the word such that $\widetilde{f}(a)$ and $\widetilde{f}(b)$ are obtained from $f(a)$ and $f(b)$, respectively, by a cyclic shift by $u$ (note that $u$ may be longer than each of the blocks $f(a)$ and $f(b)$). Then for every word $w$ we have $$u\widetilde{f}(w) = f(w)u.$$
Indeed, this is easy to see by induction on $|w|$. For the base, we have $u\widetilde{f}(a) = f(a)u$ and $u\widetilde{f}(b) = f(b)u$.
For induction step, consider $w\in \{a,b\}^*$, $\alpha\in \{a,b\}$. Then $u\widetilde{f}(w\alpha) = u\widetilde{f}(w)\widetilde{f}(\alpha) = f(w)u\widetilde{f}(\alpha) = f(w)f(\alpha)u = f(w\alpha)u$.

Next, for every word $w$ we have \begin{equation}\label{eq:f_tild}f(u)u\widetilde{f}^2(w) = f^2(w)f(u)u.\end{equation}
Indeed, applying the morphism $f$ to both sides of the equality $u\widetilde{f}(w) = f(w)u$, we obtain $f(u)f(\widetilde{f}(w)) = f^2(w)f(u)$. Now, as adding a suffix $u$ to equal words we obtain equal words, we have $f(u)f(\widetilde{f}(w))u = f^2(w)f(u)u$, which implies \eqref{eq:f_tild}. 

Setting $w=a$ and $w=b$, we conclude that the words $\widetilde{f}^2(a)$ and $\widetilde{f}^2(b)$ are obtained from $f^2(a)$ и $f^2(b)$, respectively, by a cyclic shift by $f(u)u$. And hence, the word $f^{\omega}(a) = (f^2)^{\omega}(a)$ is abelian periodic if and only if the word $(\widetilde{f}^2)^{\omega}(a)$ is abelian periodic; and the abelian period is the same.

We have  $$M_{\widetilde{f}^2}=\begin{pmatrix}
n A (nA + mB) & m A (nA + mB)\\
n B (nA + mB) & m B (nA + mB)
\end{pmatrix},$$ so $d$ is relatively prime with $\text{tr}(M_{\widetilde{f}^2})=(nA + mB)^2$ if and only if $d$ is relatively prime with  $\text{tr}(M_{f})$.

So, we suppose that $f$ is of the form \eqref{eq_f_ab}, and the word $f^{\omega}(a)$ is abelian periodic. We let $s$ denote its abelian period. Without loss of generality we set \begin{equation}\label{eq:ab_per}s = d (A + B)(nA + mB)^{t-1},\end{equation} where $t$ and $d$ are integers such that $t$ is relatively prime with $(nA + mB)$. 
Indeed, let $s'$ be an abelian period, we factorize it as $s=s_1s_2$, where $s_2$ is the maximal divisor of $s$ relatively prime with $nA+mB$, and $s_1$ is the remaining part. We then denote $s_2$ by $d$, and multiply  $s_1$ by some number to get a number of the form $(A+B)(nA+mB)^t$ for some integer $t$. We can do is as all prime divisors of $s_1$ are also divisors of  $nA+mB$. Since an abelian period multiplied by any integer is also an abelian period, we can consider a new abelian period $s$ of the form \eqref{eq:ab_per}.
By Lemma \ref{lemma:rem_d2}, we have that $\frac{s}{d} = (A + B)(nA + mB)^{t-1}$ is also an abelian period. 
\end{proof}

\subsubsection{Proof of Theorem \ref{th:main}}

Now we are ready to prove Theorem \ref{th:main}.

\begin{proof}[Proof of Theorem \ref{th:main}]

Let us first consider a primitive morphism $f$. If $\theta_2 \neq 0$, then, by Proposition \ref{utv:utv_Arb_1}, the word  $f^{\omega}(a)$ is abelian periodic if and only if $f$ is of the form $f(a) = a(ba)^k$, $f(b) = b(ab)^m$ for some integers $k$ and $m$.

    Now suppose that $\theta_2 = 0$. This means that $f$ satisfies Condition 1, i.e.
    $M_f = \begin{pmatrix}
n A & m A\\
n B & m B
\end{pmatrix}$ for some integers $A, B, m, n$ such that $(m,\, n) = 1$. By Theorem \ref{CONJ}, if $f^{\omega}(a)$ is abelian periodic (possibly with a preperiod), then there is an abelian period $p = (A + B)(nA + mB)^{K-1}$ for some $K\in \mathbb{N}$. We will assume that $K$ is such that the period $p$ is greater than the preperiod $r$ (if it is smaller, we can multiply several times by ($nA + mB$)). We have $|f^K(a)| = np$, $|f^K(b)| = mp$; consider the factorizations 
$$  f^K(a) = uv_1v_2\cdots v_{n-1}w,\quad   f^K(b) = u'v'_1v'_2\cdots v'_{m-1}w', $$
where $u$ is a preperiod, $|v_i| = |v'_j| = p$, $|uw| = p$, $|u'| = |u|$, $|w'| = |w|$, and $v_i \sim_{ab} v'_j$ for each pair of indices $i,\,j$. There are occurrences of 
$ab$ and $ba$ in the word $f^{\omega}(a)$, so $wu' \sim_{ab} v_1$ and $w'u \sim_{ab} v_1$. If there is an occurrence of $aa$ or $bb$ in $f^{\omega}(a)$, then we have, respectively, $wu \sim_{ab} v_1$ or $w'u' \sim_{ab} v_1$; both give us the factorizations we need.

If there are no occurrences of $aa$ and $bb$ in the word $f^{\omega}(a)$, then $f^{\omega}(a) = (ab)^{\omega}$ and the morphism $f$ is of the form $f(a) = (ab)^{nA},\, f(b) = (ab)^{nB}$ and the statement of the theorem holds trivially.

It remains to consider a non-primitive morphism $f$. A binary morphism prolongable on $a$ is non-primitive if and only if it has one of the following forms:
    \begin{itemize}
        \item $f(a) \in a^+$.
        In this case, $f^{\omega}(a) = a^{\omega}$.
        \item $f(b) = b^+$. In this case, by Theorem 5 and Theorem 6 in \cite{Wh19} (see also \cite{BFR14}), if the word $f^{\omega}(a)$ is not periodic, then the abelian complexity of  $f^{\omega}(a)$ is not bounded. However, the abelian complexity of an abelian periodic word is bounded, which completes the proof of the theorem. \qedhere
    \end{itemize}

\end{proof}

\subsection{Upper bound}
\label{subsection:upper_bound}

We now show the upper bound on the length of period in the case of pure abelian periodicity.

\begin{proof}[Proof of Theorem \ref{th:bound}]

By Theorem $\ref{CONJ}$, if the word $f^{\omega}(a)$ is abelian periodic, then there exists an abelian period equal to $$\text{gcd}(|f^K(a)|,\;|f^K(b)|)$$ for some $K\in \mathbb{N}$. In the case of abelian periodicity without preperiod, this is equivalent to $f^{K}(a)$ and $f^K(b)$ being contatenations of $n$ and $m$ abelian equivalent words of length $\text{gcd}(|f^K(a)|,\;|f^K(b)|)$, respectively. We now obtain an upper bound on $K$.

Consider the words $f^{t}(a)$ и $f^{t}(b)$ for $t = 1,\, 2,\, 3, \dots$. For each $t$, we check whether these words can be factorized into concatenations of $n$ and $m$ words, respectively, that are pairwise abelian equivalent. For that, given an integer $t$, we consider $f^{t}(a)$ (resp., $f^{t}(b)$) as a concatenation of $n$ (resp., $m$) words of equal length: $f^{t}(a) = u_1^tu_2^tu_3^t\dots u_n^t$ (resp., $f^{t}(b) = v_1^tv_2^tv_3^t\dots v_m^t$). Note that the definition of our morphism implies that $|u_i^t| = |v_j^t| = (A+B)(nA+mB)^{t-1}$.

We define a \textit{configuration} to be a tuple $$(x_1,\,x_2,\dots,\,x_{n-1},\, y_1,\,y_2,\,\dots,\,y_{m-1})$$ of length $n + m -2$ corresponding to the factorizations $f^{t}(a) = u_1^tu_2^tu_3^t\dots u_n^t$ and $f^{t}(b) = v_1^tv_2^tv_3^t\dots v_m^t$ in the following sense: $x_i = (1, k)$ if $u_{i+1}$ begins with $f(a)_k$, and $x_i = (2, k)$ if $u_{i+1}$ begins with $f(b)_k$ (if we factorize $f^t(a)$ into the blocks $f(a)$ and $f(b)$ as $f^t(a) = f(f^{t-1}(a))$). Analogously, $y_i = (1, k)$ if $v_{i+1}$ begins with $f(a)_k$, and $y_i = (2, k)$ if $v_{i+1}$ begins with $f(b)_k$ (if we factorize $f^t(b)$ into blocks). 

Consider the sequence of configurations $(c_t)_{t\geq 1}$ corresponding to the $t$-th power of the morphism. First, note that there are only finitely many distinct configurations, namely $(|f(a)| + |f(b)|)^{m+n-2}$. Second, observe that if applying $f$ at some step changes some configuration $\alpha$ to some configuration $\beta$, then applying $f$ always replaces $\alpha$ with $\beta$. In other words, if  $c_t=\alpha$ and $c_{t+1}=\beta$ for some $t$, then for each $j$ such that $c_j=\alpha$ we have $c_{j+1}=\beta$.  Indeed, 
consider a factorization of the word $f^{t}(a) = u_1^tu_2^tu_3^t\dots u_n^t$ into blocks $f(a)$ and $f(b)$. Then its prefix $u_1^tu_2^t\dots u_i^t$ is a concatenation of several complete blocks and a prefix (possibly empty) of another block at the end: $u_1^tu_2^t\cdots u_i^t = f(s_i)\cdot\text{pref}(f(a_i))$ for some word $s_i$ and some letter $a_i$. Now, consider the word $f^{t+1}(a) = u_1^{t+1}u_2^{t+1}u_3^{t+1}\cdots u_n^{t+1}$. 
Consider the first $i$ elements of this factorization; for the length of the corresponding word we have $$|u_1^{t+1}u_2^{t+1}u_3^{t+1}\cdots u_i^{t+1}| = (nA + mB)|u_1^{t}u_2^{t}u_3^{t}\cdots u_i^{t}|,$$ as for each $j$ we have $$|u_j^{t+1}| = \frac{|f^{t+1}(a)|}{n} = \frac{(nA + mB)|f^{t}(a)|}{n} = (nA + mB)|u_j^t|.$$ Moreover, $$u_1^{t+1}u_2^{t+1}u_3^{t+1}\cdots u_i^{t+1} = f^2(s_i)\cdot\text{pref}(f^2(a_i)).$$
Since we have $|f^2(s_i)| = (nA + mB)|f(s_i)|$, then $$|\text{pref}(f^2(a_i))| = |u_1^{t+1}u_2^{t+1}u_3^{t+1}\cdots u_i^{t+1}| - |f^2(s_i)| =$$ 
$$(nA + mB)(|u_1^{t}u_2^{t}u_3^{t}\cdots u_i^{t}| - |f(s_i)|) = (nA + mB)|\text{pref}(f(a_i))|,$$ which is independent of $t$. Therefore, for each $i = 1,\,2,\,\cdots, \,n$, we have that the $i$-th joint in the word $f^{t+1}(a)$ is uniquely determined by the previous configuration $c_t$. The same reasoning applies to $f^{t+1}(b)$.

Thus, iterating $f$ leads to a period in the sequence of configurations. However, a configuration completely determines whether all the factors in the factorizations $f^{t}(a) = u_1^tu_2^tu_3^t\cdots u_n^t$ and $f^{t}(b) = v_1^tv_2^tv_3^t\cdots v_m^t$ are abelian equivalent, which completes the proof.
\end{proof}

\begin{corollary}
\label{collorary:col4}
    Let $\Sigma = \{a,\,b\}$ and let $f: \Sigma^*\rightarrow \Sigma^*$ be a morphism prolongable on $a$. Then the pure abelian periodicity property is decidable for the word $f^{\omega}(a)$. 
\end{corollary}

\begin{proof}
By Theorem $\ref{th:main}$, if the word $f^{\omega}(a)$ is abelian periodic, then it is either periodic (which is a decidable property \cite{PER,PERI}) or there exists an abelian period equal to $$\text{gcd}(|f^K(a)|,\;|f^K(b)|)$$ for some $K\in \mathbb{N}$. By Theorem $\ref{th:bound}$, we have an upper bound for the number $K$. Thus, to check if the word $f^{\omega}(a)$ is abelian periodic, we first need to check if it is periodic. Then, if it is not periodic, it remains to check whether the words $f^k(a)$ and $f^k(b)$ can be decomposed into concatenation of abelian equivalent words of length $\text{gcd}(|f^k(a)|,\;|f^k(b)|)$ for $k\in \{1, 2,\ldots, K_{max}\}$, where $K_{max}$ is the upper bound from Theorem \ref{th:bound}.
\end{proof}

\section{Conclusion}
\label{section:conclusion}

In this paper, we provided a necessary and sufficient condition on the abelian periodicity for binary words generated by morphisms (Theorem \ref{th:main}), and the condition can be checked algorithmically in the case of pure abelian periodicity (Theorem \ref{th:bound}). Two natural questions remain open. First, in the case of abelian periodicity with preperiod, is there an upper bound on $M$ from Theorem \ref{th:main} making our criterion algorithmic? The second question is generalization of the characterization of the abelian periodicity to non-binary alphabets.

\section*{Acknowledgements}

This work was supported by the Russian Science Foundation, project 25-21-00535. 


\bibliographystyle{plain}

\end{document}